\documentclass[preprint,12pt]{elsarticle}

\usepackage{amssymb}
\usepackage{amsmath}

\usepackage{xcolor}
\usepackage{comment}
\usepackage{bm}
\usepackage{algorithm2e}
\usepackage{algorithmic}
\usepackage{amsmath,bm}
\usepackage{subcaption}
\usepackage{comment}
\usepackage[makeroom]{cancel}
\usepackage{placeins}
\usepackage{float}
\usepackage{hyperref}
\usepackage{booktabs} 

\usepackage{amsthm} 

\theoremstyle{plain}
\newtheorem{proposition}{Proposition}

\theoremstyle{plain}

\theoremstyle{definition}

\theoremstyle{plain}
\newtheorem{remark}{Remark}
\usepackage{comment}
\journal{Journal of Computational Physics}

\begin{document}

\begin{frontmatter}



\title{VR-PIC: An entropic variance-reduction method for particle-in-cell solutions of the Vlasov–Poisson equation}


\author[mit,eth]{Victor Windhab}

\author[psi]{Andreas Adelmann}

\author[mit,psi]{Mohsen Sadr} 


\affiliation[mit]{organization={Department of Mechanical Engineering, MIT},
city={Cambridge},
postcode={MA 02139},
country={USA}}

\affiliation[eth]{organization={ETH Zurich}, country={Switzerland}}

\affiliation[psi]{
organization={Paul Scherrer Institute}, 
postcode={CH-5232},
city={Villigen},
country={Switzerland}
}

\begin{abstract}
\noindent We extend the recently developed entropic and conservative variance reduction framework \cite{sadr2023variance} to the particle-in-cell (PIC) method of solving Vlasov-Poisson equation.
We show that a zeroth-order approximation that freezes the importance weights during the velocity-space kick is stable at the expense of introducing bias. Then, we propose a correction for the weight distribution using maximum cross-entropy formulation to ensure conservation laws while minimizing the introduced bias.  In several test cases  including Sod's shock tube and Landau damping we show that the proposed method maintains the substantial speed-up of variance reduction method compared to the PIC simulations in the low signal regime with minimal changes to the simulation code.

\end{abstract}




\end{frontmatter}

\section{Introduction}
\label{sec:intro}

\noindent The kinetic theory \cite{Chapman1953,klimontovich1982kinetic} bridges the gap between continuum and molecular scales by introducing an evolution equation for the particle distribution function in the phase space, leading to a high dimensional probabilistic description of short and long-range inter-molecular interactions. Applications include rarefied gas dynamics for reentry problem \cite{moss2006orion,gallis2009kinetic}, confined plasma \cite{lanti2020orb5,sadr2022linear}, and electromechanical systems \cite{cercignani2006slow}. 
\\ \ \\
In order to solve the kinetic equations, several categories of numerical methods have been developed in the literature including direct discretization of phase space \cite{platkowski1988discrete}, moment-based methods \cite{torrilhon2016modeling,frei2023moment}, and stochastic particle methods \cite{Bird,gorji2013fokker,xu2010unified,mies2023efficient}.
As a natural treatment for the curse of high dimensionality, boundary condition, and reaction modeling, particle methods have gained significant attention in the plasma physics and rarefied gas dynamics communities thanks to the celebrated direct simulation Monte Carlo \cite{Bird1963} method for the Boltzmann and particle-in-cell method for Vlasov-type equations \cite{birdsall2018plasma}. However, particle methods suffer from the statistical noise inherent in Monte Carlo estimation of moments in the low signal limit. Mitigating noise becomes essential in resolving micro-turbulence in confined plasma \cite{nevins2005discrete,cj2021effect}, and slow electromechanical systems \cite{cercignani2006slow}. 
\\ \ \\
Several approaches have been developed in the literature to reduce the noise in prediction of macroscopic moments for particle methods. This includes the deviational ($\delta f$) methods \cite{denton1995deltaf,brunner1999collisional,homolle2007low}, Multi-Level Monte Carlo (MLMC) \cite{rosin2014multilevel}, two-weight methods \cite{sonnendrucker2015split}, filtering techniques \cite{gassama2007wavelet,ricketson2016sparse,muralikrishnan2021sparse}, correlated parallel stochastic process \cite{gorji2015variance}, and importance weights \cite{al2010excursion,al2010low}. While MLMC method \cite{rosin2014multilevel} introduces no bias in estimating moments, the obtained reduction in variance is limited to small rarefaction as the correlation between coupled levels is significantly reduced as distribution departs from equilibrium. On the other hand, the two-weight method \cite{sonnendrucker2015split}, weight-based $\delta f$ method \cite{denton1995deltaf,brunner1999collisional}, and the original importance weight method \cite{al2010excursion} suffer from numerical instability in the collisional regime. The weight-less deviational method  \cite{homolle2007low} and the parallel process method \cite{gorji2015variance} do not suffer from such instabilities, however they introduce further complexity in deployment as they require extensive modeling and changes in the base codes. Finally, filtering methods are extensively used in production codes such as ORB5 \cite{lanti2020orb5,sadr2022linear}, Warp \cite{vay2018warp} and OPAL \cite{adelmann2008opal} due to their simplicity in implementation. However, they rely on smoothing out high frequencies with the assumption that such modes are not physical leading to systematic bias \cite{jolliet2012parallel,muralikrishnan2021sparse}.
\\ \ \\
\noindent  Recently, the entropic variance reduction method has been devised to overcome the numerical instability of the exact importance weight process for the Boltzmann \cite{sadr2023variance} and Fokker-Planck equation \cite{sadr2023varianceFP}. In this method, first a stable solution to the weights process is established, e.g. using Kernel Density Estimation method, at the expense of introducing bias. Then, the introduced bias is minimized by maximizing the  cross-entropy between the smoothed out weight distribution and the target posterior distribution, while enforcing the conservation of moments. The relaxation of moments are either found exactly from the unstable weight process, or analytically if available. 
\\ \ \\
In this paper, we extend the application of entropic variance reduction method \cite{sadr2023variance,sadr2023varianceFP} and devise a scheme for the particle-in-cell method of solving the Vlasov-Poisson equation. First, we show that by keeping the weights constant during the kick in the velocity space, we can obtain a stable weight process that is exact at equilibrium. In order to minimize the introduced bias as the particle distribution departs from equilibrium, we deploy the maximum cross-entropy formulation (MxE) to update weight distribution matching the target semi-analytical relaxation of moments while minimizing the introduced bias. Here, we show that matching only up to the second-order moment is enough to obtain a high accuracy in prediction and significant speedup. The proposed method inherits the numerical advantages of the importance weights methods, i.e. with the minimal change in the base code we are able to achieve orders-of-magnitude speed-up which increases quadratically as the signal magnitude reduces.
\\ \ \\
\noindent The remainder of this paper is organized as follows. First, in the method section \S~\ref{sec:methods}, we review the Vlasov-Poisson equation, its underlying dynamics, and the importance weight method. Then, in \S~\ref{sec:eq_as_cont} we discuss the choice of the control variate, derive conservative weight process during streaming and kick with the help of maximum cross-entropy formulation in \S~\ref{sec:stream_weight}-\ref{sec:weight_process_kick}. In \S~\ref{sec:sol_alg}, we provide a pseudocode for the proposed variance-reduced method. In the results section \S~\ref{sec:results}, we demonstrate the efficiency of our method in simulating 1D Sod's shock tube and 2D Landau Damping test cases in the small signal regime. Finally, in \S~\ref{sec:conclusion}, we provide a conclusion and outlook.

\section{Methods}
\label{sec:methods}

\noindent In this work, we consider the collision-less and self-consistent Vlasov-Poisson kinetic equation of the form
\begin{flalign}
    \frac{\partial f}{\partial t} + v \cdot \nabla_x f  + \frac{q}{m} E \cdot \nabla_v f = 0.
    \label{eq:kinetic_vlasov}
\end{flalign}
Here, $f( v, x ;t)$ denotes the single-particle distribution function in the phase space with velocity $ v\in \mathbb{R}^d$ and position $ x\in \mathbb{R}^d$ at time $t \in \mathbb{R}^+$, $q$ is the elementary electron charge, $m$ the mass of each particle, and $E = -\nabla_x \phi$ is the electric field. The self-consistent electrostatic potential $\phi$ is the solution to the Poisson equation of the form
\begin{flalign}
- \nabla^2 \phi(x;t) =  \rho_0 - \rho (x;t)
\label{eq:poisson}
\end{flalign}
where $\rho_0$ denotes a constant ion charge density introduced to enforce quasi-neutrality and $\rho(x;t):=q n(x;t)$ is the electric charge density with $n(x;t):=\int f(v,x;t) dv$ denoting the electron number density.
\\ \ \\
In particle-in-cell methods, the particle distribution function $f$ is discretized using computational particles with random variable $X$ and $V$ for particle position and velocity, i.e.
\begin{flalign}
    f(v,x;t) 
    &\approx \frac{1}{\delta \Omega}\sum_{i=1}^{N_p} \delta(X^{(i)}-x) \delta(V^{(i)}-v)
\end{flalign}
where 
$\delta  \Omega$ denotes the average phase space volume occupied by each particle to satisfy normalization $\int f dv dx = N_\text{p,tot}$ giving the total number of particles in the system \cite{birdsall2018plasma} unless mentioned otherwise.
\\ \ \\
In the PIC loop, first particle charges are \textit{scattered} (deposited) onto a spatial grid to estimate the right hand side of \eqref{eq:poisson}.  Although there is a wide range of shape functions used in practice for the scatter step, for simplicity we use a Dirac delta deposition scheme in this work. Then, the electrostatic potential $\phi$ is estimated on the grid by numerically solving the Poisson equation \eqref{eq:poisson} using finite difference, finite elements, or spectral methods. For simplicity, here we use 2nd order finite differences to discretize the Poisson equation. Then, $E$  is estimated at the particle position $X$ using a form of interpolation known as \textit{gather}. Here, we use piecewise constant estimator to gather $E$ at particle positions. Finally, the underlying Vlasov dynamics
\begin{flalign}
    dV(t) &= \frac{q}{m} E( X(t);t) dt
    \\
   \textrm{and} \ \ \ \  dX(t) &= V(t) dt,
    \label{eq:vlasov_system}
\end{flalign}
can be simulated by splitting the operators, and deploying Euler, the Leap Frog or other numerical schemes to integrate $X$ and $V$ in time. We refer to the first operator as \textit{kick}, and the second as \textit{streaming}. Although the proposed variance reduction method is independent of the choice of method for the operator splitting or the discretization scheme, for simplicity in this study we consider the first-order semi-implicit Euler scheme, i.e. for time interval $t\in[t_0, t_0+\Delta t)$
\begin{flalign}
    V(t_0+\Delta t) &= V(t_0) + \frac{q}{m} E(X(t_0);t_0) \Delta t
    \\
 \textrm{and} \ \ \ \    X(t_0+\Delta t) &= X(t_0) + V(t_0+\Delta t) \Delta t~.
    \label{eq:vlasov_system_discrete}
\end{flalign}

\subsection{Review of the variance reduction method}

\noindent The idea of variance reduction (VR) with the importance weight is to relate the velocity moments of the target non-equilibrium  distribution to the ones of the control variate distribution via
{\small
\begin{flalign}
\int R(  v) f(  v,  x;t) d     v 
=& \underbrace{\int  R(  v) \left(1-w(  v,  x; t)\right) f(v, x ;t) d     v}_{I_1:=} 
\nonumber \\
& + \underbrace{\int  R(  v) f^\mathrm{cont.}(  v, x; t) d     v}_{I_2:=},
\label{eq:VR_formulatiton}
\end{flalign}
}
where the importance weights $w$ is defined as 
\begin{flalign}
w(  v,  x; t) := \frac{f^\mathrm{cont.}(  v,  x; t)}{f(  v,  x;t)}.
\label{eq:weight_def}
\end{flalign}
Here $ R(  v)$ denotes a velocity polynomial, e.g. $ R(  v) \in \{ 1,v_i,v_iv_j...\}$ for $i,j, ...=1,2,3$. 
Given $  V^{(1)},...,  V^{(N_p)}$ samples of the distribution function $f(  v,   x; t)$, the VR method uses the finite number of particles to estimate the integral $I_1$ instead of directly computing $\int R f dv$ with particles. Therefore, the variance-reduced estimate can be computed simply via
{\small
\begin{flalign}
n_\mathrm{VR}\ \Big \langle  R(  V) 
\Big \rangle_\mathrm{VR}
&:= 
n \langle (1-W)R\rangle  + \underbrace{\int  R(  v) f^\mathrm{cont.}(  v,  x; t) d     v}_{\mathrm{analytical\ computation}},
\label{eq:momR_VR}
\\
\langle (1-W)R\rangle &\approx \frac{\sum_{i=1}^{N_p}   (1-W^{(i)}) R(  V^{(i)}) }{N_p}.
\end{flalign}
}

\noindent \sloppy Here, $W$ is the random variable associated with weight $w$, and we deploy the notation $\langle R\rangle\approx \int R fdv/\int fdv$ to denote the Monte Carlo approximation of the velocity moment and  $\langle.\rangle_\mathrm{VR}$ the variance reduced estimate.
Since the moments of $f^\mathrm{cont.}$ are known analytically, the error/noise in evaluating eq.~\eqref{eq:VR_formulatiton} is only associated with computation of $I_1$. Given that the control variate  $f^\mathrm{cont.}$ is close to $f$, the Monte Carlo integration of $I_1$ is more efficient than $\int  R(  v) f(  v,  x; t) d v$ since $|1-w(  v,   x; t)|\ll 1$, leading to a variance-reduced estimation of the moments. 

\subsection{Equilibrium as the control variate}
\label{sec:eq_as_cont}

\noindent The variance reduction method leads to speed-up when the non-equilibrium distribution $f$ is close to the control variate distribution function $f^\mathrm{cont.}$. For problems in the small signal regime, the particle distribution $f$ is close to the local equilibrium distribution function, i.e. Maxwell-Boltzmann distribution function 
\begin{flalign}
f^\mathrm{eq}(  v,  x;t) = \frac{n(  x; t)}{\left(2 \pi \theta^2(  x; t)\right)^{d/2}} \exp \left( - \frac{||  v -   U(  x;t)||_2^2}{2 \theta^2(  x; t)} \right),
\label{eq:feq_loc}
\end{flalign}
\sloppy which provides a reasonable control variate  $f^\mathrm{cont.}$, as it is the equilibrium solution to the Boltzmann and Vlasov-Poisson equation. Here, $  U(x;t)=\int   v f(v,x;t) d     v/\int f(v,x;t) dv$ denotes the bulk velocity, $\theta(x;t)=\sqrt{k_b T(x;t)/m}$ the thermal velocity, $k_b$ the Boltzmann constant, $m$ the mass of each particle, and $ T(x;t)=m/(d n k_b) \sum_{i=1}^d   \int (v_i- U_i(x;t))^2 f d     v$ the temperature. 
\subsection{Weight process during advection}
\label{sec:stream_weight}

\noindent Although the local equilibrium provides a control variate close to $f$, we note that the variation of its moments in $x$ and $t$ complicates the advection of particles in the physical space. Therefore, we also define a global equilibrium as
\begin{flalign}
f^{\mathrm{eq},0}(  v) = \frac{n_0}{\left(2 \pi \theta^2_0\right)^{d/2}} \exp \left( - \frac{||  v-  U_0||_2^2}{2 \theta^2_0} \right),
\label{eq:feq_glob}
\end{flalign}
where $n_0$, $\theta_0$ and $  U_0$ denote the corresponding global equilibrium moments that are independent of space and time. These parameters are typically determined in the beginning of the simulation from the problem specification as they define the global equilibrium of the entire system.
\\ \ \\
\noindent Now, we can avoid the complication associated with evolving the importance weights of the local equilibrium in the particle advection by transforming weights to the global equilibrium distribution Eq.~(\ref{eq:feq_glob}) before applying the advection step using    
\begin{flalign}
\hat{w}(  v,  x;t)=w(  v,  x;t)  \frac{f^{\mathrm{eq},0}(  v)}{f^{\mathrm{eq}}(  v,  x; t)}~.
\label{eq:loc_to_global}
\end{flalign}
Here, we denote variable represented in the global control variate framework by $(\hat{.})$,  i.e. $\hat{w}=f^{\mathrm{eq},0}(  v)/f(  v,x;t)$. Since the global equilibrium $f^{\mathrm{eq},0}(  v)$ is invariant with respect to time and space, it can be shown that
\begin{proposition}
The importance weights \eqref{eq:weight_def} in the global frame under the spatial-independent global equilibrium \eqref{eq:feq_glob} as the control variate are conservative on the particle trajectory during streaming in the physical space, i.e.
\begin{flalign}
\frac{D\hat{w}}{Dt}\Big|_\mathrm{streaming}&=0.
\end{flalign}
\end{proposition}
\begin{proof}
    It is trivial to show that
    \begin{flalign}
         \frac{D(f^{\mathrm{eq},0})}{Dt} &=   \frac{D(\hat wf)}{Dt}
        \\
          & =  \hat w \frac{D(f)}{Dt} + f \frac{D(\hat w)}{Dt}
         \\
         &= f \frac{D(\hat w)}{Dt}
    \end{flalign}
    since ${Df}/{Dt}|_\mathrm{streaming}=0$ is the material derivative as the exact operator during the streaming of particles. 
    Given the global equilibrium has no variation in space and time, it follows  ${Df^{\mathrm{eq},0}}/{Dt}|_\mathrm{streaming}=0$, and in turn ${D\hat w}/{Dt}=0$. 
\end{proof}
\noindent Once particles are streamed, the inverse of this map
\begin{flalign}
w(  v,  x;t)=\hat{w}(  v,  x;t)  \frac{f^{\mathrm{eq}}(  v,  x; t)}{f^{\mathrm{eq},0}(  v)}~
\label{eq:global_to_local}
\end{flalign}
can be used to bring weights back to the local equilibrium framework.

\subsection{Weight process during kick}
\label{sec:weight_process_kick}

\noindent Consider the particles inside the cell near x in the time interval $t\in[t_0, t_0+\Delta t]$. The evolution of the particle distribution function during the kick in the velocity space can be written as
\begin{flalign}
    \frac{\partial f(v,x;t)}{\partial t}\Big|_\mathrm{kick} = - \frac{q}{m} E(x;t_0) \cdot \nabla_v f(v,x;t) 
    \label{eq:vlasov-kick}
\end{flalign}
The underlying process for the particle velocity follows
\begin{flalign}
    d V_t = \frac{q}{m} E (X_t; t_0) dt,
    \label{eq:vel_process}
\end{flalign}
which is known as the \textit{kick} in the context of particle-in-cell methods. By solving this ODE numerically for $t\in[t_0, t_0+\Delta t]$, we can find $V(t_0+\Delta t)$.  What are the particle weights after the kick $W(t_0+\Delta t)$?
\begin{proposition}
\label{prop:post-kick-weight-estimate}
(Zeroth order approximation) If the equilibrium distribution $f^\mathrm{eq}$ does not evolve during the kick, the importance weight of particles remains constant on the particle trajectory in the velocity space.
\end{proposition}

\begin{proof}
 Let us inspect the weight definition after the kick at $t_0+\Delta t$, i.e.
\begin{flalign}
    W_{t_0+\Delta t} := \frac{f^\mathrm{eq.}(V_{t_0+\Delta t},x;t_0+\Delta t)}{f(V_{t_0+\Delta t},x;t_0+\Delta t)}.
\end{flalign}
Since PIC resolves the non-equilibrium distribution directly by solving \eqref{eq:vel_process}, we have
\begin{flalign}
f(V_{t_0+\Delta t},x;t_0+\Delta t) \approx  f(V_{t_0},x;t_{0})
\end{flalign}
on the particle trajectory in the velocity space using the method of characteristics. Therefore, the post-kick weight can be simplified to
\begin{flalign}
    W_{t_0+\Delta t} &\approx 
    \frac{f^\mathrm{eq.}(V_{t_0+\Delta t},x;t_0+\Delta t)}{f(V_{t_0},x;t_0)}
    \\
    &= 
    \frac{f^\mathrm{eq.}(V_{t_0},x;t_0)}{f(V_{t_0},x;t_0)} \frac{f^\mathrm{eq.}(V_{t_0+\Delta t},x;t_0+\Delta t)}{f^\mathrm{eq.}(V_{t_0},x;t_0)}
    \\
    &= W_{t_0} \frac{f^\mathrm{eq.}(V_{t_0+\Delta t},x;t_0+\Delta t)}{f^\mathrm{eq.}(V_{t_0},x;t_0)}.
\end{flalign}
Assuming that the local equilibrium is conserved during the kick, we can find
\begin{flalign}
    W_{t_0+\Delta t} &\approx W_{t_0}.
    \label{eq:fixed_W}
\end{flalign}
We refer to eq.\eqref{eq:fixed_W} as the \textit{zeroth order approximation}.
\end{proof}
\noindent Although the kinetic energy is conserved during the particle kick in the velocity space with the electric field $E$, the Vlasov-Poisson equation introduces a local source term at the cell near $x$ in the momentum equation during the kick that contradicts the assumption in the proposition \ref{prop:post-kick-weight-estimate}. Therefore, the naive use of \eqref{eq:fixed_W} in estimating the post-kick weights introduces bias since the local equilibrium is not conserved during the kick.
\\ \ \\
Furthermore, the explicit particle solution to the Vlasov-Poisson equation does not strictly enforce conservation laws for a finite number of particles. Next, we derive the post-kick moments for the conservation laws, and illustrate how maximum cross-entropy (MxE) can be used to enforce them on discrete particles.

\begin{proposition}
\label{prop:balance_moments}
The kick in the Vlasov equation introduces a source term in the local momentum balance
\begin{flalign}
    \int v_i  \frac{\partial}{\partial t}   f(v,x;t) dv \Big|_\mathrm{kick} &= \frac{q}{m} n(x; t) E_i(x;t),
\end{flalign}
near the position $x$ at time $t$, while mass and kinetic energy are conserved.
\end{proposition}

\begin{proof}
By multiplying eq.~\eqref{eq:vlasov-kick} by $R\in \{ 1, v_i, \xi^2\}$, where $\xi_i = v_i - U_i(x;t)$, $\xi^2 = \sum_j \xi_j^2$ and $i,j=1,...,d$, and taking integrals in velocity space, using the integration by part, divergence theorem, and assuming $f,\nabla_vf,...\rightarrow 0$ faster than any polynomial $v_i,v_iv_j,...\rightarrow \infty$, we obtain the balance equation for the mass

\begin{flalign}
    \int \frac{\partial}{\partial t} f(v,x;t) dv \Big|_\mathrm{kick} &= - \frac{q}{m}E(x;t) \cdot \int \nabla_v f(v,x;t) dv   \nonumber \\
    &=0,
\end{flalign}
momentum
{\small
\begin{flalign}
    \int v_i \frac{\partial}{\partial t}   f(v,x;t) dv \Big|_\mathrm{kick} &= - \frac{q}{m} E_j (x;t)  \int {v_i} \nabla_{v_j} f(v,x;t) dv
    \nonumber \\
    &=  \frac{q}{m} E_j(x;t) \int \delta_{ij}  f(v,x;t) dv \   \nonumber
    \\
    &=  \frac{q}{m} n(x; t) E_i(x;t),
\end{flalign}
}

\noindent and energy
{\footnotesize
\begin{flalign}
    &\int  \xi_k \xi_k \frac{\partial}{\partial t}  f(v,x;t) dv \Big|_\mathrm{kick}\nonumber \\
    &= - \frac{q}{m} E_j(x;t) \left( \int v_k v_k \nabla_{v_j} f(v,x;t) dv
-2 U_k(x;t) \int v_k \nabla_{v_j} f(v,x;t) dv 
+ U_k(x;t)U_k(x;t)  \int \nabla_{v_j} f(v,x;t) dv
    \right) \nonumber 
    \\
    &= - \frac{q}{m} E_j(x;t)
    \left( 
    - \int 2 v_j  f(v,x;t) dv
    +
    2U_j(x;t) \int f(v,x;t) dv
    \right) \nonumber
    \\
    &= - \frac{1}{m} E_j(x;t)
    \left( 
    - 2 \rho U_j(x;t)
    +
    2\rho U_j(x;t)
    \right)
    \nonumber \\
    &= 0.
\end{flalign}}
\end{proof}

\noindent 
Due to the mass conservation of the kick operator, i.e.
\begin{flalign}
    \rho (x;t)&=\rho (x;t_0) \nonumber
    \\
    &=\rho(x;t_0+\Delta t) \ \ \ \ 
    \forall t\in [t_0, t_0+\Delta t],
\end{flalign}
the corresponding potential $\phi$ and electric field $E$ become constant with respect to time during the kick. Therefore, the solution to the outcome system of ODEs for the balance of central moments in the Proposition~\ref{prop:balance_moments} becomes

{\footnotesize
\begin{flalign}
     &\int \begin{pmatrix}
 1  \\
   v_i   \\
  \xi_k \xi_k
\end{pmatrix}
f(v,x;t_0+\Delta t) dv
=
 \int \begin{pmatrix}
 1  \\
   v_i   \\
  \xi_k \xi_k
\end{pmatrix}
f(v,x;t_0) dv
+
 \frac{q}{m} \begin{pmatrix}
 0  \\
  n(x;t_0) E_i(x;t_0)    \\
 0   
\end{pmatrix} \Delta t~.
\label{eq:sol_ode}
\end{flalign}
}

\noindent We use this analytical expression for the evolution of moments up to energy during kick \eqref{eq:sol_ode} and estimate the variance-reduced non-equilibrium moments at $t_{0}+\Delta t$ via
\begin{flalign}
    n_\mathrm{VR} \langle R (v) \rangle_\mathrm{VR} \Big|_{t_0+\Delta t}
\approx
 n_\mathrm{VR} \langle R(v) \rangle_\mathrm{VR} \Big|_{t_0}
+
 \frac{q}{m} \begin{pmatrix}
 0  \\
  n_\mathrm{VR} {E}_{VR}    \\
 0   
\end{pmatrix} \Delta t,
\label{eq:post_kick_VR_moments}
\end{flalign}
where $n_\mathrm{VR}$ is the variance reduced estimate of number density, and $\phi_\mathrm{VR}$ the variance reduced potential computed from
\begin{flalign}
- \nabla^2 \phi_\mathrm{VR}(x,t_0) =  \rho_0 - \rho_\mathrm{VR}(x,t_0)~.
\label{eq:poisson_vr}
\end{flalign}

\noindent Having estimated $n_\mathrm{VR}$, $E_\mathrm{VR}$, and $\langle R (v) \rangle_\mathrm{VR}$
at $t_0+\Delta t$, we can now calculate moments of weighted samples $\langle W R \rangle $ using eq.~\eqref{eq:VR_formulatiton} and \eqref{eq:post_kick_VR_moments} as
\begin{flalign}
    &n_\mathrm{VR} \langle R(v) \rangle_\mathrm{VR}\Big|_{t_0+\Delta t} =  \int R(v) f^\mathrm{eq,0}(v) dv + n \langle (1-W)R(V)\rangle \Big|_{t_0+\Delta t} \nonumber
     \\
    \implies &
    n \langle W R(v) \rangle \Big|_{t_0+\Delta t} = 
     \int R(v) f^\mathrm{eq,0}(v) dv + n \langle R(v) \rangle \Big|_{t_0+\Delta t}
    - {n_{VR}} \langle R(v) \rangle_\mathrm{VR}\Big|_{t_0+\Delta t}.
    \label{eq:post_kick_weight_def}
\end{flalign}
This gives us the target moments for the weighted particles after the kick
{\small \begin{flalign}
    \mu :=  
    \frac{1}{n} \int R(v) f^\mathrm{eq,0}(v) dv + \langle R(v) \rangle \Big|_{t_0+\Delta t}
    - \frac{n_\mathrm{VR}}{n} \langle R(v) \rangle_\mathrm{VR}\Big|_{t_0}
    -\frac{q n_{VR}}{m n} \begin{pmatrix}
 0  \\
 E_{VR}    \\
 0   
\end{pmatrix} \Delta t
    \label{eq:post_kick_wmom}
\end{flalign} }

\noindent which we can enforce on post-kick particle distribution using the MxE method. 

\begin{remark}
The post-kick moment balance derived in Proposition \ref{prop:balance_moments} is exact for continuous functions.
\end{remark}

\begin{remark}
The post-kick relation for the target moments in eq. \eqref{eq:post_kick_wmom} is valid regardless of how the moments $\langle R\rangle$  are numerically computed, i.e. it holds true for any shape function in the scatter and gather. However, in this paper we consider particles as summation of Dirac delta.
\end{remark}

\noindent Considering the zeroth-order approximation of the post-kick weight described in Proposition~\ref{prop:post-kick-weight-estimate} as the prior $\mathcal{F}^\mathrm{prior}$ with weights $W^\mathrm{prior}$,
and denoting the target $M$ post-kick moments $\mu\approx \langle W R(v) \rangle$ at $t_0+\Delta t$ in Eq.~\eqref{eq:post_kick_wmom} as the target moments,
the extremum of the cost functional
\begin{flalign}
C[\mathcal{F}] := \int \mathcal{F} \log\left(\mathcal{F}/\mathcal{F}^\mathrm{prior}\right) d     v 
& + \sum_{i=1}^M \lambda_i   \left(\int R_i  \mathcal{F} d     v-\mu_i \right)
\label{eq:cross_MED_cost}
\end{flalign}
provides us with the maximum cross-entropy distribution function (MxE)  \cite{debrabant2017micro}, i.e.
\begin{flalign}
\delta C/\delta \mathcal{F} = 0 \implies 
\mathcal{F}^* = \mathcal{F}^\mathrm{prior} \exp\left(\sum_{i=1}^M \lambda_i R_i \right).
\label{eq:calF}
\end{flalign}
Here, $\delta C/\delta \mathcal F$ denotes the  Gateaux derivative of cost functional $C[\mathcal F]$ with respect to distribution function $\mathcal{F}$ \cite{giaquinta2013calculus}. Using the definition of weights eq.~\eqref{eq:weight_def}, we can find the update rule for the weights as
\begin{flalign}
W^* = W^\mathrm{prior} \exp\left(\sum_{i=1}^M \lambda_i R_i  \right).
\label{eq:calW}
\end{flalign}
\noindent The Lagrange multipliers, $  \lambda$, appearing in \eqref{eq:calW} can be computed by following the Newton-Raphson approach. As formulated in \cite{debrabant2017micro},  the unconstrained dual formulation $D(  \lambda)$ provides us with the gradient $  g = \nabla D(  \lambda) $ and Hessian $  H(  \lambda)=
\nabla^2 D(  \lambda)$ as
\begin{flalign}
g_i &= \mu_i - \int R_i  \mathcal{F} ^\mathrm{prior} \exp\left(\sum_k \lambda_k R_k  \right) d     v\ \ \ \textrm{for}\ i=1,...,M
\label{eq:gradient_ME}
\\
H_{i,j}&= - \int R_i R_j \mathcal{F}^\mathrm{prior} \exp\left(\sum_k \lambda_k R_k (  v) \right) d     v \ \ \ \textrm{for}\ i,j=1,...,M
\label{eq:Hessian_ME}
\end{flalign}
leading to the iterative scheme
\begin{flalign}
  \lambda^{(k+1)} =   \lambda^{(k)} -    H^{-1}(  \lambda^{(k)})   g(  \lambda^{(k)}),
\label{eq:lambda_update}
\end{flalign}
where $k$ indicates the iteration index. In Algorithm~\ref{alg:ME_optimization}, we present the Newton-Raphson algorithm of solving MxE optimization problem given target moments $\mu$ and prior particle distribution.
\\ \ \\
{\small 
\begin{algorithm}[H]
 \caption{The Newton-Raphson algorithm to solve the MxE optimization problem with $N_p$ weighted samples of the distribution function with prior weights $ W^{\mathrm{prior},\ (i)}\ \text{for} \ i=1,...,N_p$ matching M target  polynomial $   R(   v)$ moments denoted by $   \mu$ with the tolerance of $\epsilon$. Here we use the convention that $\mu_1$ is the zeroth order moment (density).}
\SetAlgoLined
-Initialize $W^{(i)}= W^{\mathrm{prior},\ (i)}$ for $i=1,..,N_p$\;
 \While{$||\sum_{i=1}^{N_p} W^{(i)}    R(   V^{(i)})/N_p-   \mu    ||_1/||    \mu||_1>\epsilon$}{
-Estimate gradient $ g_{l-1} = \mu_l - \sum_{i=1}^{N_p} R_l(   V^{(i)}) W^{(i)}/N_p$ for $l={2},...,M$\;
-Estimate Hessian $ H_{k-1,l-1} =  - \sum_{i=1}^{N_p} R_k(   V^{(i)})R_l(   V^{(i)}) W^{(i)}/N_p$  for $k,l={2},...,M$\;
-Solve the ${(M-1)\times (M-1)}$ linear system $   H    \lambda^* =    g$\; 
-Update the estimate of Lagrange multipliers  $   \lambda \leftarrow    \lambda -    \lambda^*$\;
 -Update weights $W^{(i)} = W^{\mathrm{prior},\ (i)} \exp\left(\sum_{{k=2}}^M \lambda_{k-1} R_k (   V^{(i)}) \right)  $  for $i=1,...,N_p$\;
 -Enforce mass conservation $W^{(i)} = W^{(i)} \left(\mu_{1}/\sum_{j=1}^{N_p} W^{(j)}\right)$  for $i=1,...,N_p$\;
 }
 \label{alg:ME_optimization}
\end{algorithm}
}
\ \\
\begin{remark} Since $\mathcal{F}^*$ in eq.~\eqref{eq:calF} is the unique solution to \eqref{eq:cross_MED_cost} that minimizes the cross-Shannon entropy from the prior, it follows directly
\begin{flalign}
\sum_i^{N_p} W^{*, (i)} \log\left(W^{*, (i)}/W^{\mathrm{prior},\ (i)}\right) \leq 
\sum_i^{N_p} W^{(i)}\log\left(W^{ (i)}/W^{\mathrm{prior},\ (i)}\right)
\end{flalign}
for any weight distribution $W$ that satisfies the semi-analytical moment constraints eq.~\eqref{eq:post_kick_wmom}.  This inequality ensures stability of the weight process in the cross-entropy sense.
\end{remark}

\begin{remark}
The introduced variance reduction method preserves the positivity of the particle weights. This is due to the fact that the transformation maps between local and global equilibrium \eqref{eq:loc_to_global}-\eqref{eq:global_to_local}, as well as the maximum cross-entropy formulation \eqref{eq:calW} strictly preserve the positivity of the initial weights.
\end{remark}

\begin{remark}
\label{remark:ill_cond_MxE}
The proposed VR-PIC method can be generalized to any $M$ moments as far as the highest order term is even. Although the Hessian matrix \eqref{eq:Hessian_ME} becomes ill-conditioned for polynomial basis functions as $M\rightarrow \infty$  in the limit of realizability \cite{hauck2008convex}, MxE is well-conditioned for application in variance reduction as the distribution remains close to equilibrium.
\end{remark}

\subsection{VR-PIC algorithm}
\label{sec:sol_alg}

\noindent 
In this section, we put together all essential elements of VR-PIC to build its solution algorithm. Besides the basic operations of the PIC method including scatter charges to a grid, solve Poisson equation on the grid, and gathering the electric field on the particle position, the proposed VR method introduces a few extensions. First, particles need to carry a weight that maintains the correlation to equilibrium distribution function. Second, the weights need to be mapped between the global and local equilibrium before and after the kick step.  Third, after the kick, the particle weights need to be corrected give the semi-analytical post-kick moments.
In Algorithm~\ref{alg:ME-VP}, we present the overall pseudo-code of the VR-PIC. We emphasize that the proposed method does not affect the dynamics of the non-equilibrium distribution function. The attached evolution of particle weight can be seen as a parallel process that allows variance reduction in the prediction of moments leaving the underlying dynamics untouched.
\\ \ \\
{\small
\begin{algorithm}[H]
 \caption{Algorithm for the VR-PIC solution to Vlasov-Poisson eq. using importance weights. Here, $N_\mathrm{cell}$ indicates number of cells in physical space, $N_{p,\mathrm{tot}}$
 denotes total number of particles,  and $t_\mathrm{final}$ the total simulation time.}
\SetAlgoLined
 -Initialize particles in the phase space\;
 \While{$t<t_\mathrm{final}$}{
 -Scatter particles on a grid and estimate $\rho$ and $\rho_\mathrm{VR}$\;
 -Solve Poisson eq. \eqref{eq:poisson} for $\phi$\;
  -Solve Poisson eq. \eqref{eq:poisson_vr} for $\phi_\mathrm{VR}$\;
 -Compute electric field $E$ and $E_\mathrm{VR}$\;
  -Map weights from global to local equilibrium using Eq.~\eqref{eq:global_to_local}\;
 \For{$i=1,...,N_\mathrm{cells}$}{
 -Kick particles in the velocity space solving Eq. \eqref{eq:vlasov-kick}\;
 -Compute the post-kick weighted moments \eqref{eq:post_kick_wmom}\;
 -Perform the ME optimization Algorithm~\ref{alg:ME_optimization}\;
}
 -Map 
 weights
 from local back to global equilibrium Eq.~\eqref{eq:loc_to_global}\;
  -Stream particles $ X^{(i)}= X^{(i)}+ V^{(i)} \Delta t$ for $i=1,...,N_{p,\mathrm{tot}}$\;
  -Apply boundary conditions\;
  -Compute moments for post-processing\;
  -$t=t+\Delta t$\;
 }
 \label{alg:ME-VP}
\end{algorithm}
}

\section{Results}
\label{sec:results}

\subsection{Sod's Shock Tube}

    

\noindent
A wide range of instabilities in collision-less confined plasma is related to gradient-driven modes such as drift waves \cite{horton1999drift} and Ion Temperature Gradient (ITG) \cite{cole2021tokamak}.
As a relevant test case, we consider the one-dimensional ($d=1$) Sod shock tube problem \cite{sod1978survey} in order to investigate the accuracy and computational savings of the proposed VR-PIC solution to the Vlasov-Poisson system of equations against the PIC approach. Consider a solution domain in physical space with the relevant direction in $x_1 \in [0,L]$. The other spatial dimensions are assumed to be large enough so that one only needs to simulate the dynamics in $x_1$ direction. The domain is discretized into $N_{x_1} = 50$ cells of size $\Delta x_1= 0.02$. We uniformly initialize the position of particles of mass $m=1$ and charge $q=1$ in $x_1$ such that a discontinuity in the number density is formed at $x_1=L/2$, i.e.  
\begin{flalign}
\rho(x_1,t=0) =
\left\{\begin{matrix}
 \rho_0(1+\alpha/2)& \ \ \text{for}\ x_1<L/2\\ 
\rho_0(1-\alpha/2)& \ \ \text{otherwise.}
\end{matrix}\right. 
\label{eq:init_rho_sod}
\end{flalign}
\noindent Here, we consider $\rho_0 = 1$, and $\alpha:=\Delta \rho/\rho_0=(\rho_\mathrm{left}-\rho_\mathrm{right})/\rho_0$ which denotes the initial perturbation in the charge density. Initially, we distribute particles uniformly in the left and right side of the initial discontinuity given the corresponding densities. We sample the velocity of particles from the Maxwell-Boltzmann equilibrium distribution function with initial temperature $T( x_1, t=0)=T_0$, set $k_bT_0/m=1$, and the initial bulk velocity of $ U_0:=U( x_1, t=0)= (0,0,0)^T$. In order to simplify analysis, we consider specular reflection as the boundary condition for the particles leaving the domain and study the solution before the shock and rarefaction waves reach the boundary.
\\ \ \\
\noindent The evolution of the PIC and the proposed VR-PIC method is done as follows. The particle properties are scattered onto the grid using Nearest Grid Point (NGP) deposition to approximate the charge density $\rho(x;t)$ and form the right hand side of the Poisson eq. \eqref{eq:poisson}. 
 We discretize the Poisson equation using a second-order central finite difference scheme and solve the resulting tridiagonal linear system directly in $O(N_{x_1})$ complexity and $O(1)$ extra memory \cite{teukolsky1992numerical,golub2013matrix}. Here, we apply homogeneous Neumann boundary conditions for the potential, i.e. $\nabla_x \phi = 0$ at the domain boundaries $x=0$ and $L$, which ensures a vanishing electric field at the walls consistent with global charge neutrality. All simulations for this test case were performed on a single core and thread of Intel Xeon Gold 6150 CPU. 




\subsubsection{Stability}

\noindent A necessary requirement for the proposed entropic variance-reduction (VR) method is the stability of the prior particle weight process. As discussed in \S~\ref{sec:sol_alg}, we consider the zeroth-order approximation for the weight process on the particle trajectory during the kick as the prior, see Proposition.~\ref{prop:post-kick-weight-estimate}, and deploy MxE to correct the bias in the posterior while imposing semi-analytical moment constraints \eqref{eq:post_kick_VR_moments}.
We test the stability of the weight process by simulating VR-PIC solution to the Sod's shock tube problem with the initial perturbation of $\alpha=0.01, 0.2, 0.4$ and $0.8$ using in total $N_{p,\mathrm{tot}}=2\times 10^{5}$ particles, and deploy the time step size of $\Delta t=0.002$ for $t\in[0,t_f]$ where $t_f=1$. 
In Fig.\ref{fig:weight_over_time}, we show the evolution of the weights in $||.||_\infty$-norm with increasing signal strengths $\alpha$ as a function of time equipped with and without MxE correction. Here, we also show the average number of iterations taken for the MxE algorithm \ref{alg:ME_optimization} to converge with the tolerance of $\epsilon=10^{-8}$. This result clearly indicates a controlled weight growth for VR-PIC with and without MxE correction, specifically in the low signal regime as expected.

\begin{figure}[H]
    \centering \hspace*{-2.2cm}
    \includegraphics[width=1.3\linewidth]{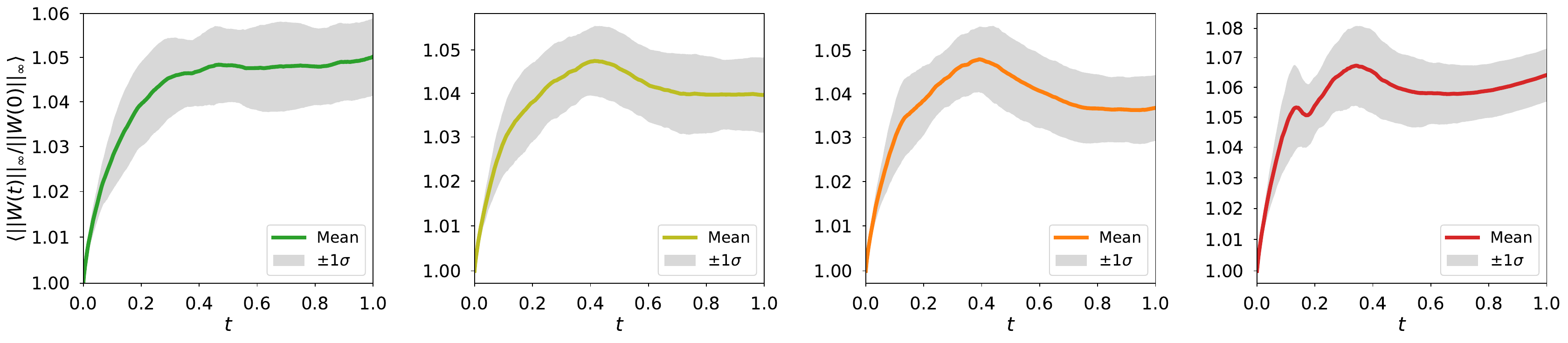}
    \\ \vspace{0.5cm}
    \hspace*{-2.2cm} \includegraphics[width=1.3\linewidth]{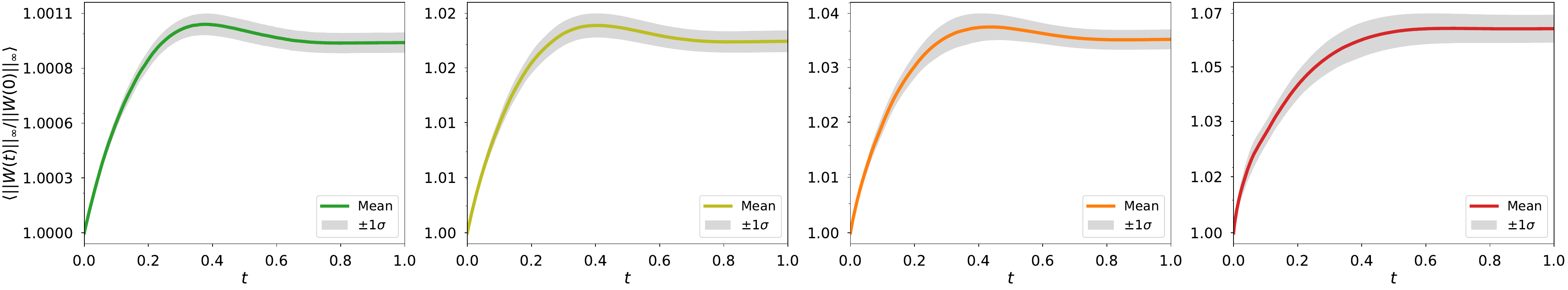}
    \\ \vspace{0.5cm}
   \hspace*{-2cm}  \includegraphics[width=1.3\linewidth]{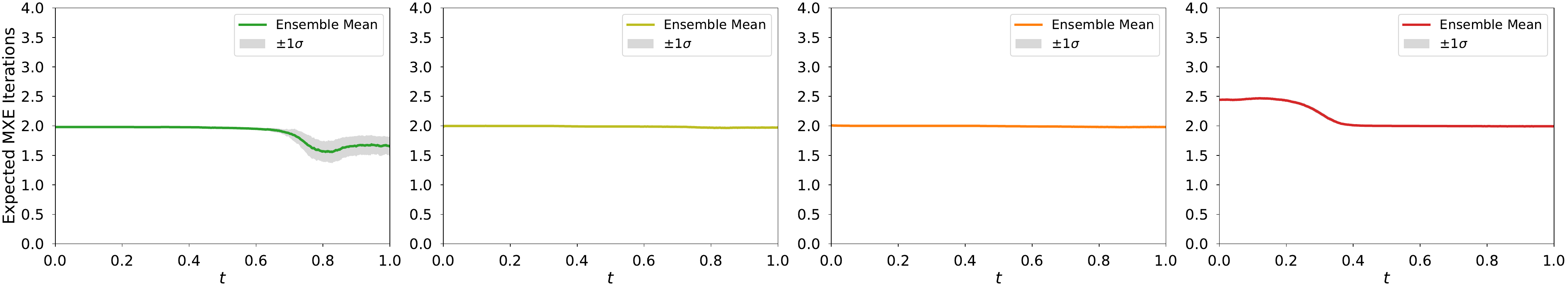}
    \caption{Normalized maximum particle weight $\langle \lVert W(t) \rVert_\infty / \lVert W(0) \rVert_\infty \rangle$ as a function of time $t$ for increasing perturbation amplitudes $\alpha$ for VR-PIC without MxE (top) and with MxE (middle), as well as the average number of iterations needed for the MxE algorithm \ref{alg:ME_optimization} to converge. Solid lines represent the mean over 100 ensembles with the respective shaded areas indicating the standard deviation ($\pm 1\sigma$).
    }
    \label{fig:weight_over_time}
\end{figure}

\subsubsection{Consistency}


\noindent Next, we investigate the accuracy of proposed VR-PIC against a reference PIC solution.
For this comparison, first we generated a \textit{low-noise} reference solution by averaging the results of a PIC simulation of shock-tube with $\alpha=0.01$ and $0.2$ over $3\times  10^6$ and $2\times 10^5$ ensemble, respectively. Then, we compute the VR-PIC solution averaged over only $100$ and $2000$ ensembles in order to achieve a similar uncertainty in prediction as the reference solution. 
\\ \ \\
In Fig.~\ref{fig:sod_profiles}, we show the evolution of the normalized density and temperature profiles at $t/\Delta t= 10, 30, 50, 70$ time steps. The black markers represent the high-ensemble PIC reference solution, whereas the solid blue and dashed red lines show the results of the low-ensemble VR-PIC simulation with and without MxE correction. 
Since the VR-PIC algorithm without MxE correction does not strictly ensure conservation laws, see Proposition~\ref{prop:balance_moments}, its solution slightly deviates from the reference solution. However, as expected, the solution of the proposed VR-PIC method equipped with MxE correction captures the reference solution with a high accuracy.

\begin{figure}[H]
    \centering \hspace*{-2cm}
    \includegraphics[width=1.3\linewidth]{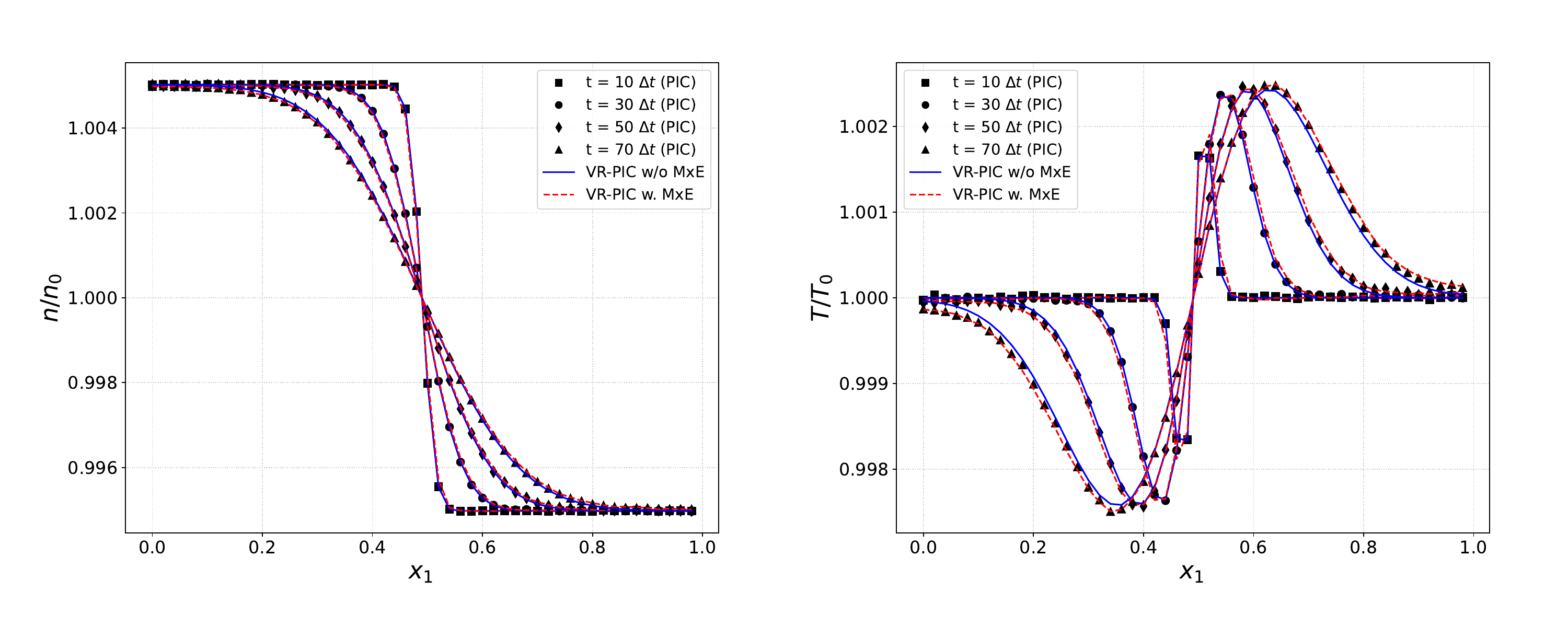}
    \\ \hspace*{-2cm}
    \includegraphics[width=1.3\linewidth]{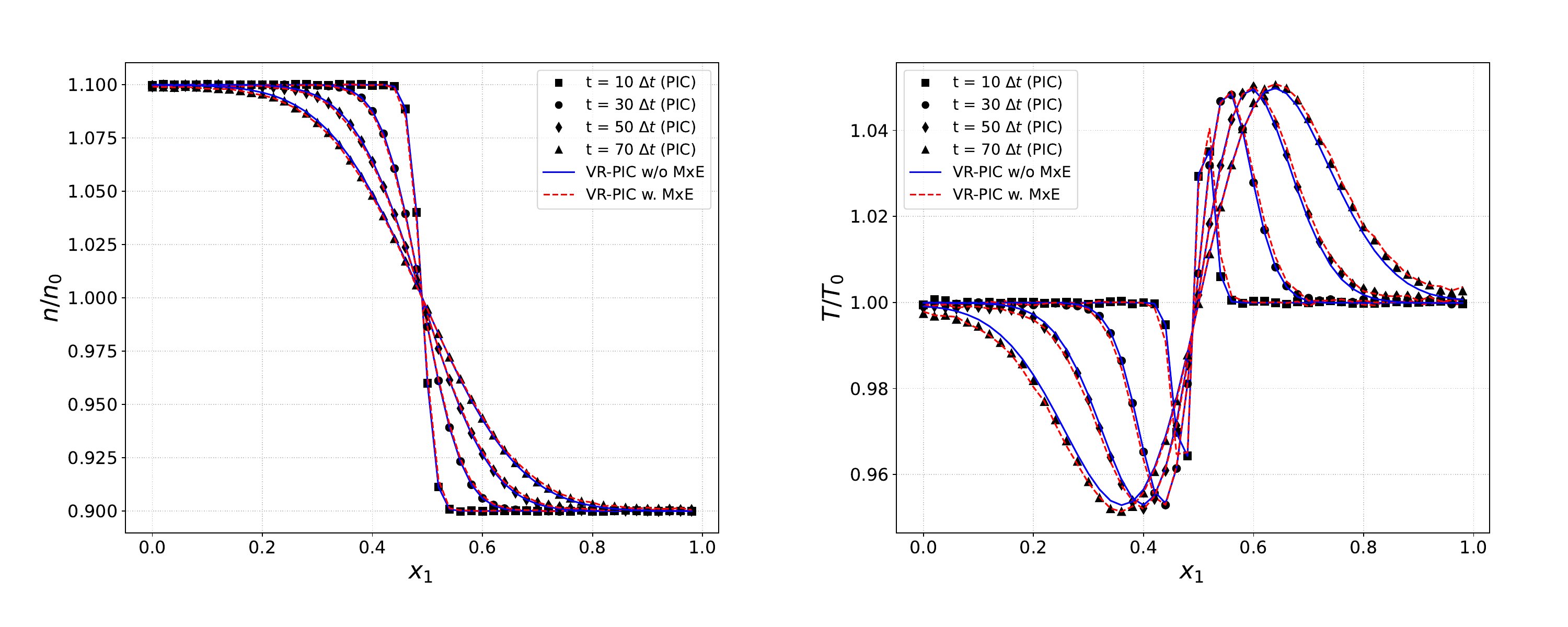}
    \caption{Normalized number density ($n/n_0$) and temperature ($T/T_0$) profiles for the Sod's shock tube problem at different times for initial perturbation of $\alpha=0.01$ (top) and $0.2$ (bottom). The black markers represent the reference solution obtained using the PIC simulation averaged over $3 \times 10^6$ and $2 \times 10^5$ ensembles, respectively. The VR-PIC  solution with (red) and without  maximum cross-entropy correction (blue) are obtained using 100 and 2000 ensembles, respectively. 
    }
    \label{fig:sod_profiles}
\end{figure}


\subsubsection{Computational Saving}

\noindent Since the solution of the Vlasov-Poisson equation to the Sod's shock tube test case with initial perturbation $\alpha=0.01$ and $\alpha=0.2$  obtained with VR-PIC method equipped with MxE correction using only $100$ ensemble is matching the reference PIC using $10^6$ and $2\times 10^3$ ensemble respectively at a high accuracy and similar noise level, the VR-PIC method achieves almost one-to-four order of magnitude reduction in the computational cost with negligible overhead. In Table~\ref{tab:computational_cost}, we report the execution time, as well as other cost related parameters. Clearly, as the signal strength $\alpha$ reduces in magnitude, VR-PIC method becomes more efficient. Next, we investigate this trend.



\begin{table}[ht]
    \centering
    \caption{Comparison of averaged computational cost between the reference PIC and the proposed VR-PIC solution for the 1D Sod's shock tube problem.}
    \label{tab:computational_cost}
     \resizebox{\textwidth}{!}{
    \begin{tabular}{c c c c c c}
        \toprule
        $ \alpha$ &
        {Method} &${N_{p,\mathrm{tot}}}$ & {Ensembles} & {Runtime/Ens. [s]} & {Total Runtime} [min] \\
        \midrule
        0.01 &
        PIC  & $\approx 2 \times 10^5$ & $3 \times 10^6$ & 0.18 & $\approx$ 9000.0 \\
        & VR-PIC          & -- & $10^2$    & 0.53 & $\approx$ 0.9 
        \\
        \midrule
        0.2 &
        PIC  & $\approx 1 \times 10^5$ & $2 \times10^5$ & 0.011 & $\approx$ 36.7  \\
        & VR-PIC          & -- & $2 \times 10^3$    & 0.026 & $\approx$ 0.9 
        \\
        \bottomrule
    \end{tabular}
    }
\end{table}

\subsubsection{Further Cost Analysis}
\label{sec:const_analysis}

\noindent The primary advantage of the variance reduction technique is most apparent at the low-signal limit. In order to better quantify this, we investigate its performance as a function of the initial signal strength, i.e. the initial perturbation in the charge density $\alpha$ for the Sod's shock tube test case. For this analysis, we performed PIC and VR-PIC simulations
across a range of $\alpha$ and measured the variance relative to a high-ensemble reference solution for each case. 
\\ \ \\
Here, we calculate the relative variance by computing the ratio of $L_2^2$-norm on the difference between the low-ensemble and reference density profiles $||\rho-\rho^\mathrm{ref}||_2^2$ to the $L_2^2$-norm of the reference profile signal $||\rho^\mathrm{ref}||_2^2$. In Figure~\ref{fig:error_vs_signal}, we show the relative variance in prediction as a function of the signal strength $\alpha$. As expected, the relative variance of PIC increases quadratically as the signal magnitude $\alpha$ reduces. Furthermore, as expected, the proposed solution of the VR-PIC method inherits the celebrated feature of VR methods, i.e. the relative variance remains constant.
\\ \ \\
\noindent This result is a central validation of the efficiency of the proposed method. This demonstrates that the computational cost of the VR-PIC method is decoupled from the magnitude of the physical signal being measured, i.e. here $\alpha$. For any given error tolerance, the PIC method requires a quadratically increasing number of ensembles (and thus cost) to resolve a progressively weaker signal. However, the VR method achieves the same accuracy, providing a significant computational speed-up that quadratically increases as the problem enters the challenging low-signal regime.

\begin{figure}[H]
    \centering
    \includegraphics[width=0.55\linewidth]{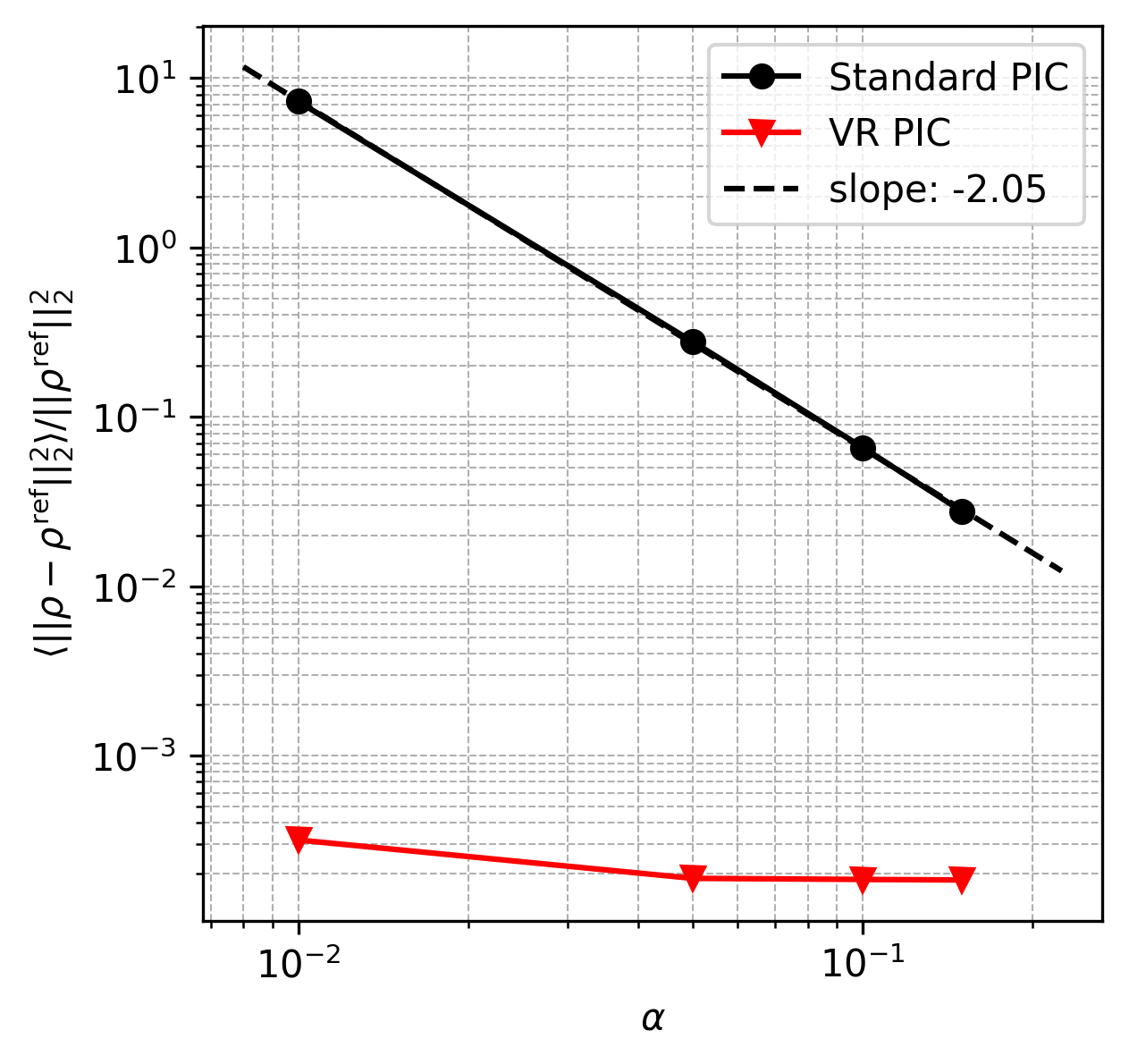}
    \caption{Relative statistical error in the density profile as a function of the initial signal strength $\alpha$. The PIC simulation error (black) increases quadratically as the signal $\alpha$ reduces, while the uncertainty of the VR-PIC solution (red) remains constant and orders-of-magnitudes smaller than the one for PIC. The data points were averaged over 100 ensembles and compared to a reference solution obtained with $5\times 10^4$ ensembles.}
    \label{fig:error_vs_signal}
\end{figure}


\subsection{Landau Damping}

\noindent In this section, we showcase the application of the introduced variance-reduced particle-in-cell method by simulating Landau Damping in $d=2$ dimensions \cite{birdsall2018plasma,chen2011energy} in the small signal regime. Using the acceptance/rejection method, we initialize the particle position and velocity following 
\begin{flalign}
    f(v;x,t=0) &= \exp(-|v|^2/2) (1 + \alpha \cos(kx_1) )(1 + \alpha \cos(kx_2)) / Z 
    \label{eq:landau_damping_dist}
\end{flalign}
where $Z$ is the normalizer such that $\int f(v;x,t) dv dx = L^2$. Here, we consider $\Omega_x=[0,L]^2$ as the physical space, i.e. $x \in \Omega_x$ where $L=4\pi$, deploy periodic boundary conditions for the particles and the Poisson equation, and set $\alpha=0.05$ and $k=0.5$. The periodic boundary condition allows us to solve the Poisson equation using the fast Fourier transformation. 
\\ \ \\
\sloppy We discretize the physical space $\Omega$ with a $100\times 100$ uniform mesh, use time steps size satisfying $\theta_0 \Delta t/\Delta x=0.05$ to respect an estimate of the Courant–Friedrichs–Lewy (CFL) condition, where $\theta_0=1$ is the considered normalized initial thermal velocity. We simulate the evolution of discretized particle distribution and compute the moments using the PIC and the developed VR-PIC methods with MxE for the moment correction with the tolerance of $\epsilon=10^{-6}$ in the Algorithm~\ref{alg:ME_optimization}. The implementation of the proposed entropic variance reduction method in C++/CUDA is available on Github\footnote{\href{https://github.com/mohsensadr/VRPIC}{https://github.com/mohsensadr/VRPIC}}. 
Here we use the NVIDIA A100 GPU with 40 GB of memory.
\\ \ \\
In Fig.~\ref{fig:landau_n}-\ref{fig:landau_rho_T_on_line}, we show the PIC and VR-PIC solution to the normalized number density $n$, the electrostatic potential $\phi$, and electric field $E = -\nabla \phi$, and temperature $T$ at $t/\Delta t = 50, 100, 150$ and $200$ steps using $10^6$ particles. Here, we also make a comparison to a reference PIC solution obtained using $ 10^8$ particles, since in practice often using more particles instead of ensemble averaging is preferred to obtain more accurate and less noisy solution. As expected, the VR-PIC solution delivers significantly smaller noise in prediction compared to the PIC while matching the noise level and accuracy of the expensive reference PIC solution gaining about significant speedup with minor overhead.

\begin{figure}[H]
  	\centering      \hspace*{-2cm}\begin{tabular}{ccccc}
 \rotatebox{90}{\footnotesize \hspace{1cm} ref. PIC  \hspace{1.3cm} VR-PIC \hspace{1.8cm} PIC} & 
   \includegraphics[scale=0.35]{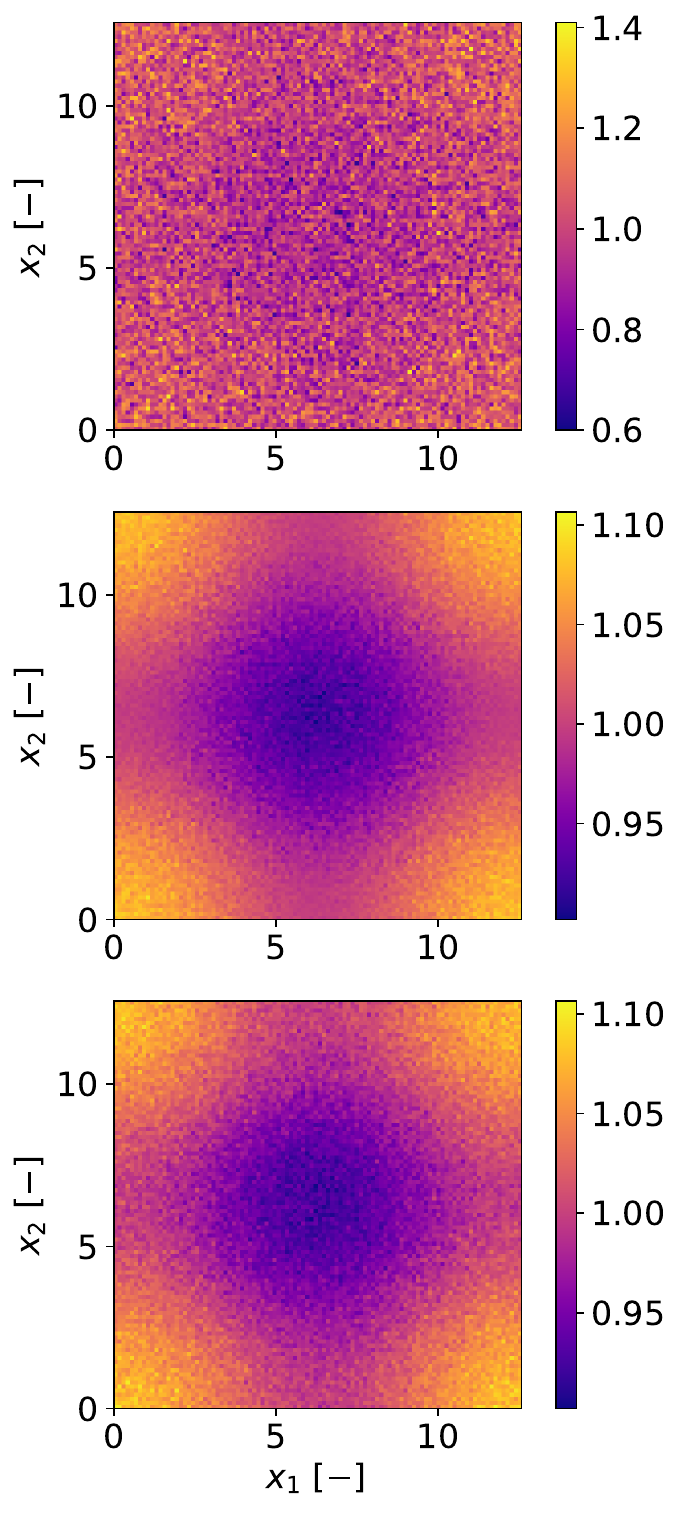}
   &
   \includegraphics[scale=0.35]{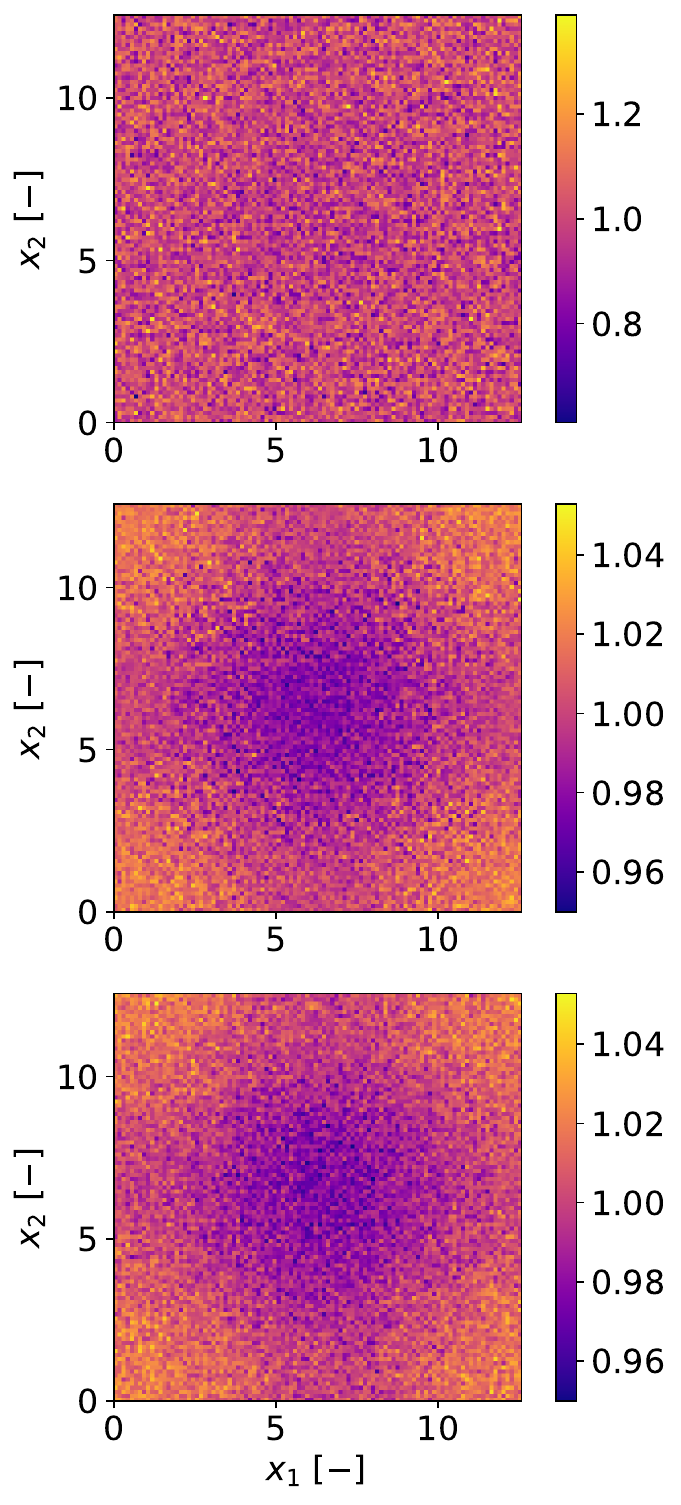}
   &
   \includegraphics[scale=0.35]{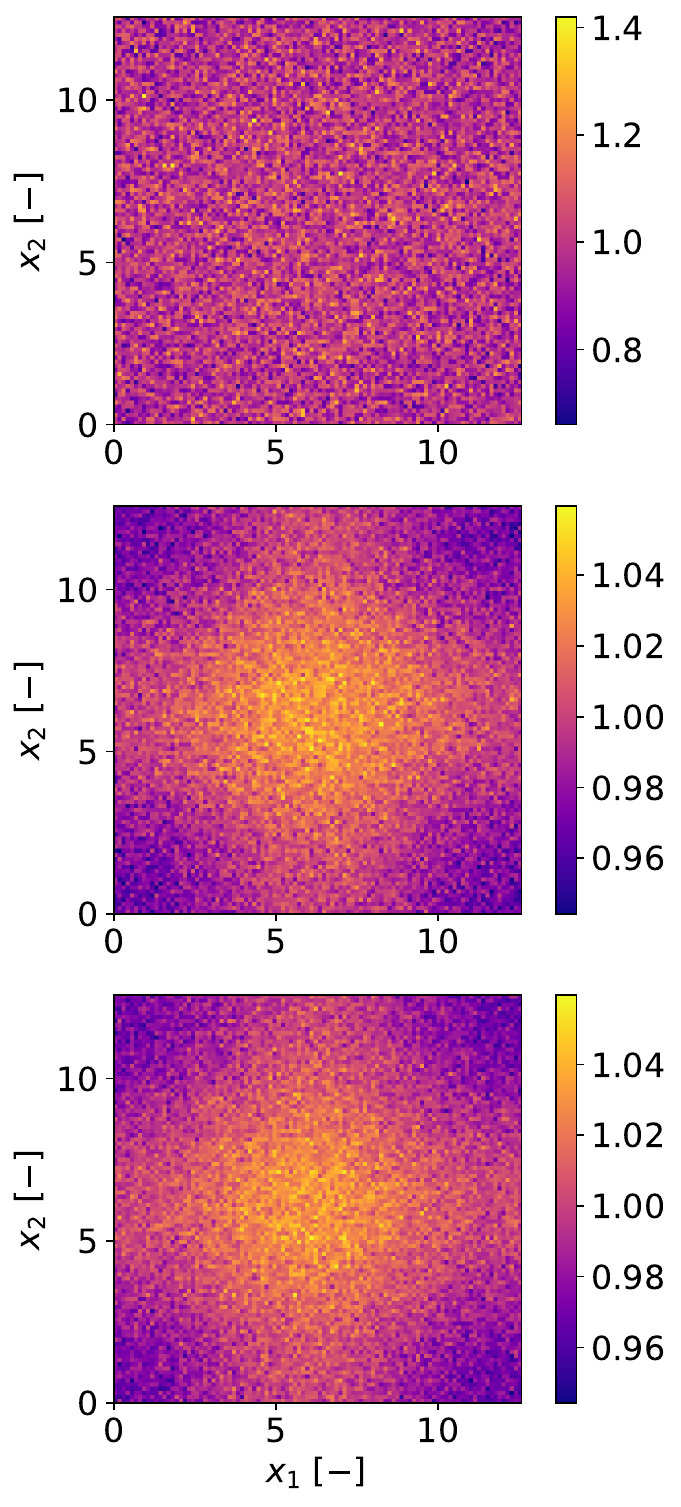}
   &
   \includegraphics[scale=0.35]{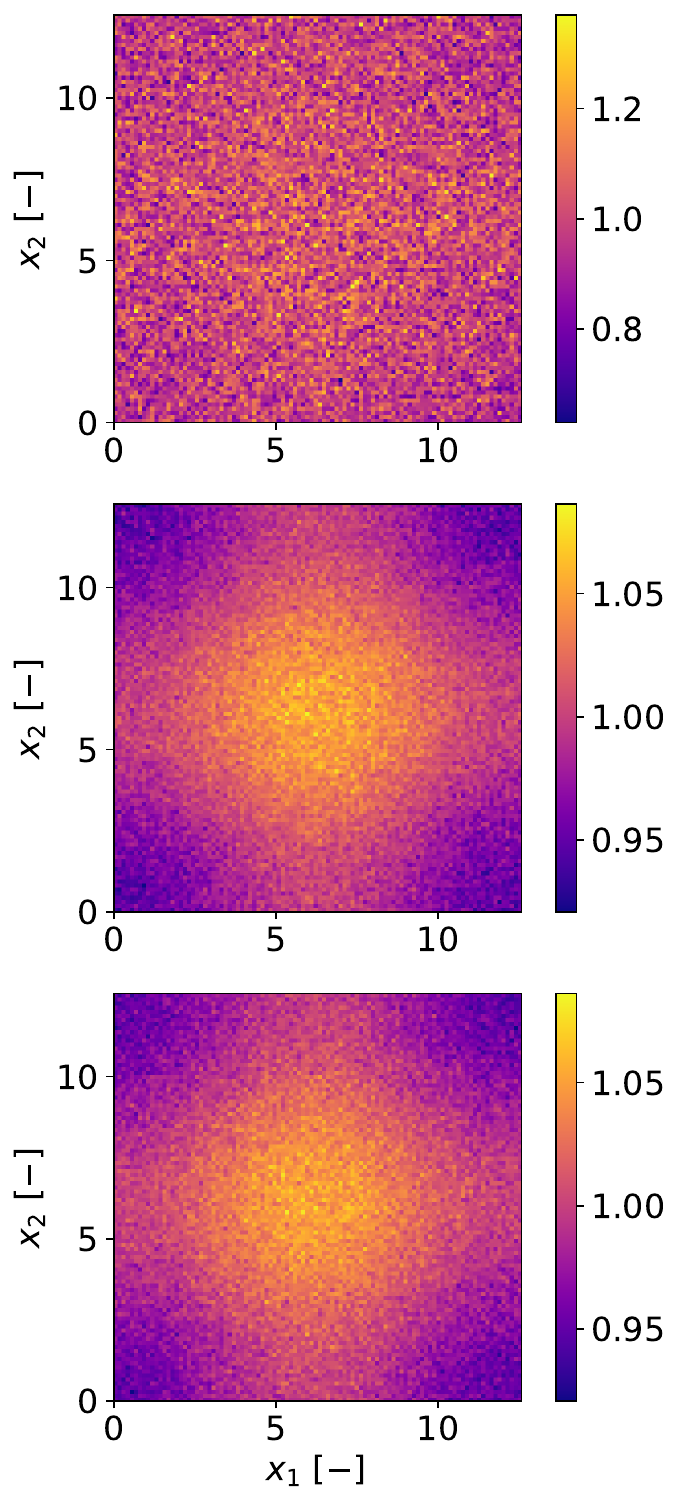}
  \\ 
  &
  (a) $t=50 \Delta t$ 
  & 
  (b)  $t=100 \Delta t$
  &
 (c) $t=150 \Delta t$ 
  & 
 (d) $t=200 \Delta t$
  \end{tabular}
  \caption{
  Evolution of the normalized charge density $\rho/\rho_0$ for the Landau Damping test case computed using PIC (top) and VR-PIC (middle) method with $10^6$ particles as well as PIC with $10^8$ as the reference (bottom), respectively, at $t/\Delta t=50, 100, 150$ and $200$ time steps.
  }
  \label{fig:landau_n}
\end{figure}

\begin{figure}[H]
  	\centering     
     \hspace*{-2cm}\begin{tabular}{ccccc}
 \rotatebox{90}{\footnotesize \hspace{1cm} ref. PIC  \hspace{1.3cm} VR-PIC \hspace{1.8cm} PIC}
 & 
   \includegraphics[scale=0.35]{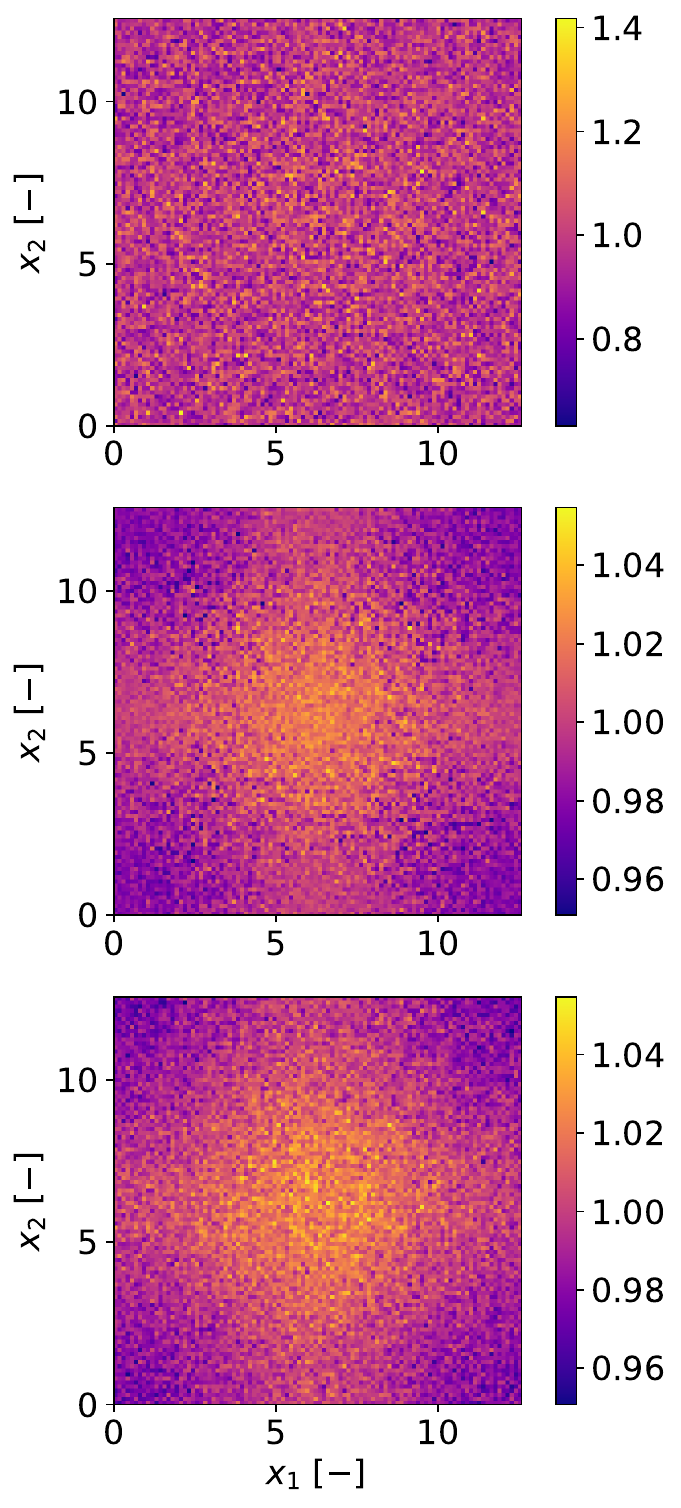}
   &
   \includegraphics[scale=0.35]{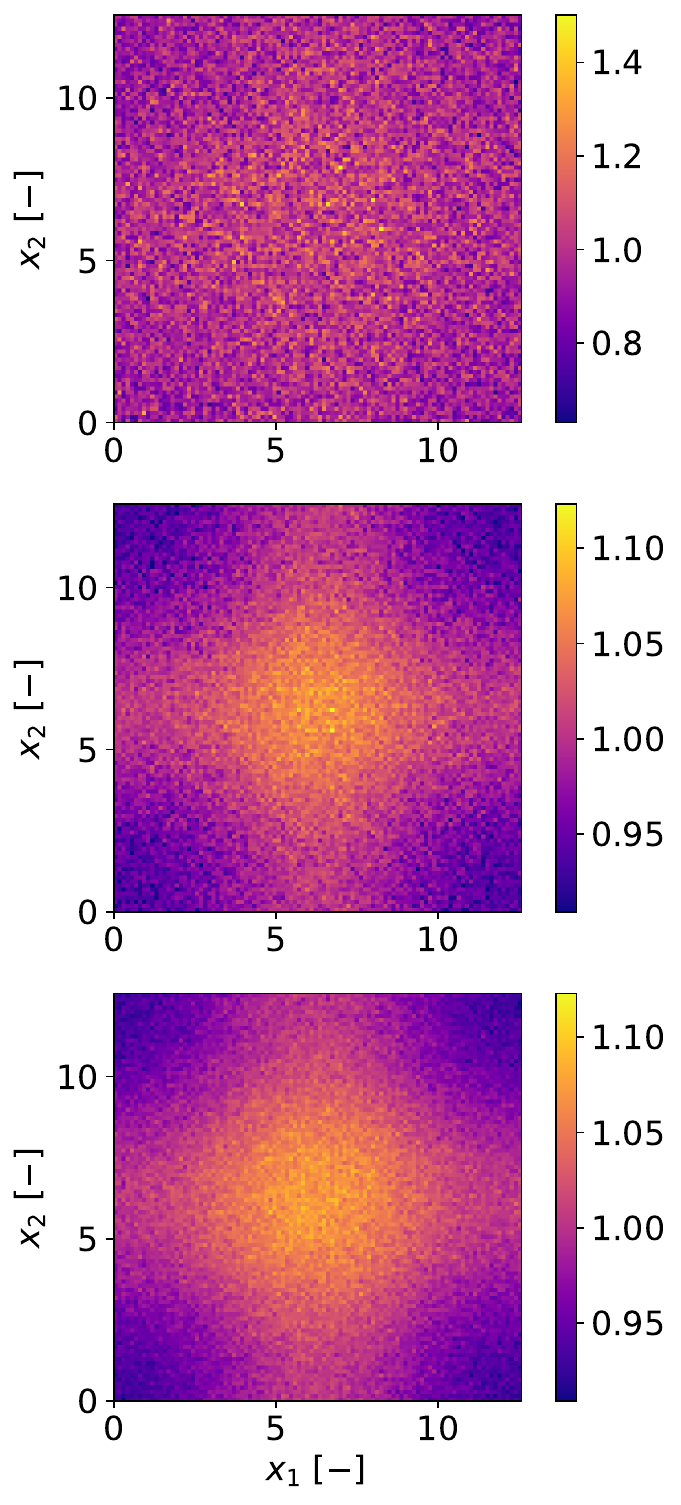}
   &
   \includegraphics[scale=0.35]{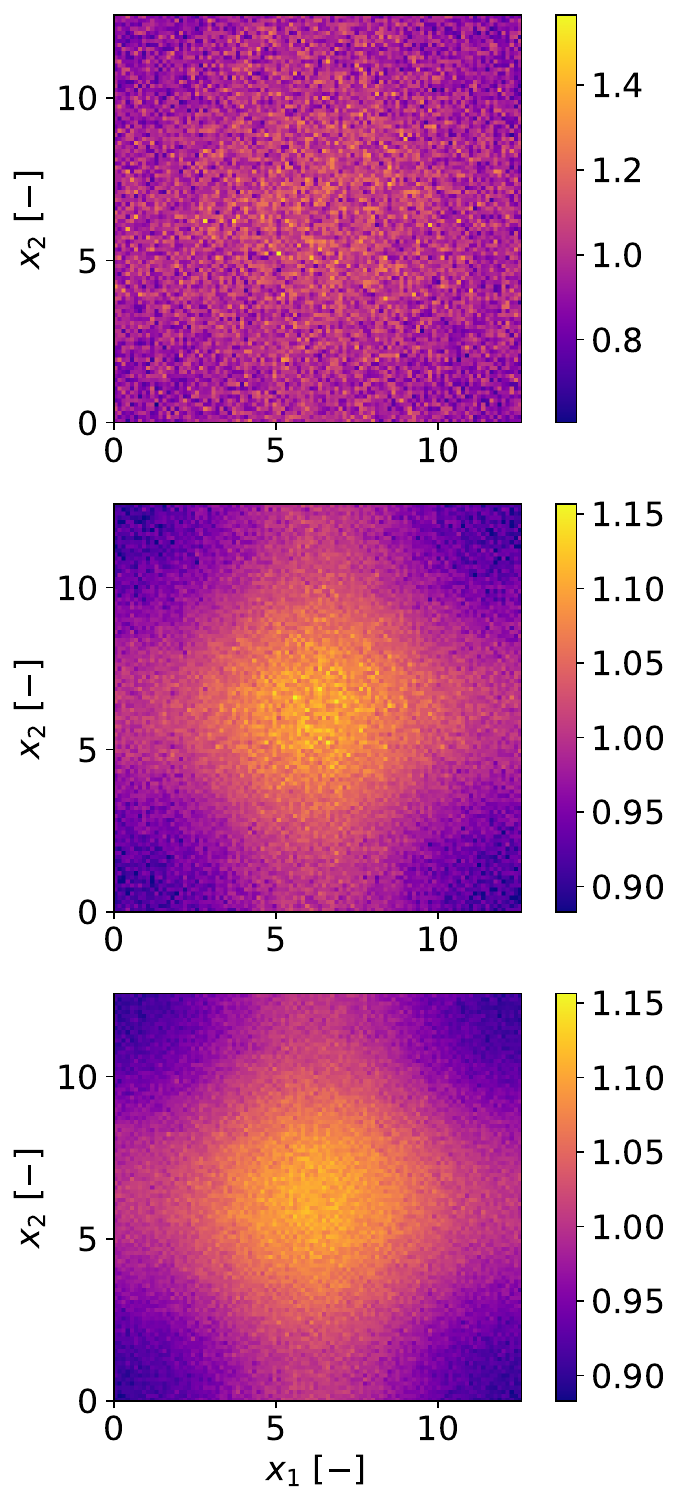}
   &
   \includegraphics[scale=0.35]{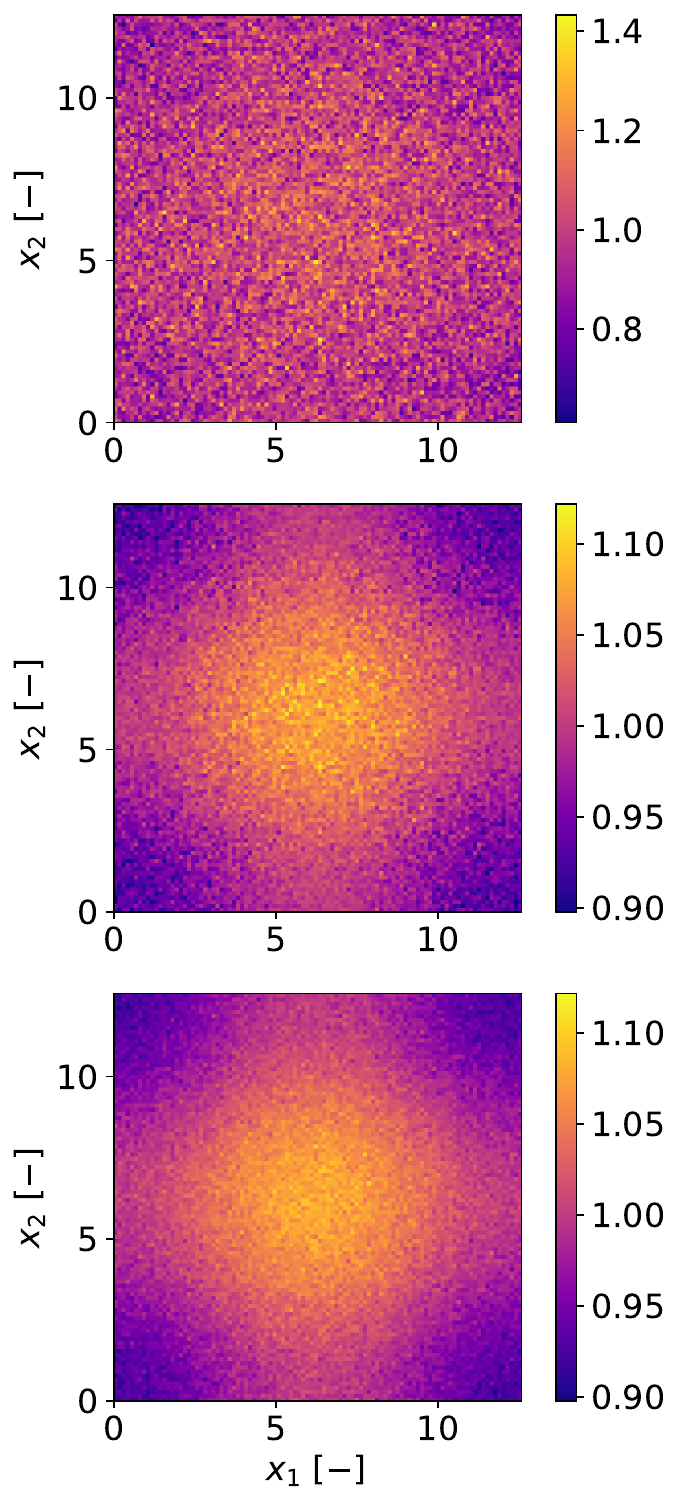}
  \\
  &
  (a) $t=50 \Delta t$ 
  & 
  (b)  $t=100 \Delta t$
  &
 (c) $t=150 \Delta t$ 
  & 
 (d) $t=200 \Delta t$
  \end{tabular}
  \caption{
  Evolution of the normalized temperature $T/T_0$ for the Landau Damping test case computed using PIC (top) and VR-PIC (middle) method with $10^6$ particles as well as PIC with $10^8$ as the reference (bottom), respectively, at $t/\Delta t=50, 100, 150$ and $200$ time steps.
  }
    \label{fig:landau_T}
\end{figure}

\begin{figure}[H]
\hspace*{-2.5cm}
    \begin{tabular}{cccc}
   \includegraphics[scale=0.5]{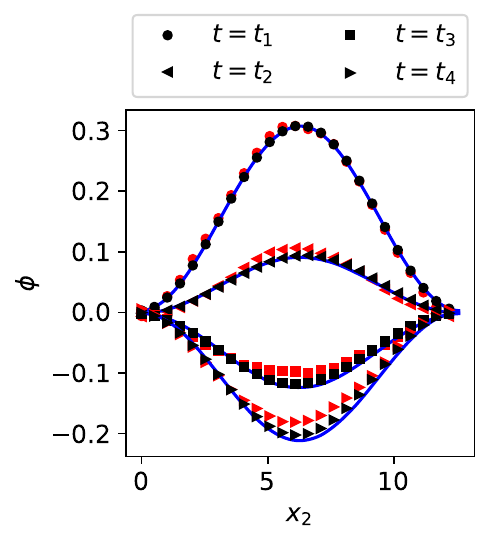}
   &
  \includegraphics[scale=0.5]{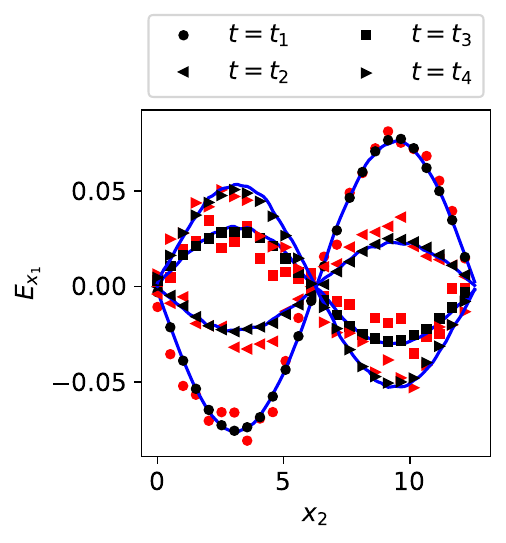}
   &
   \includegraphics[scale=0.5]{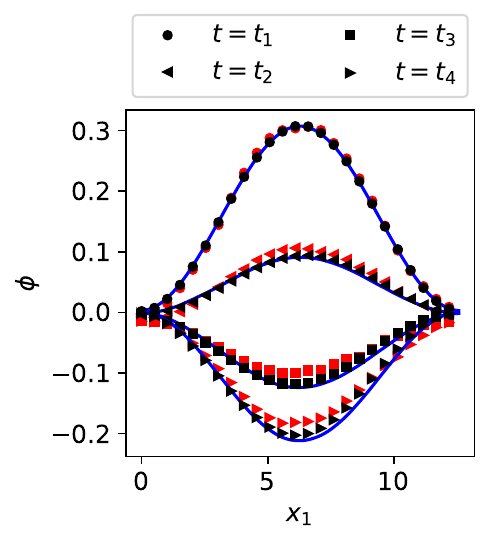}
   &
   \includegraphics[scale=0.5]{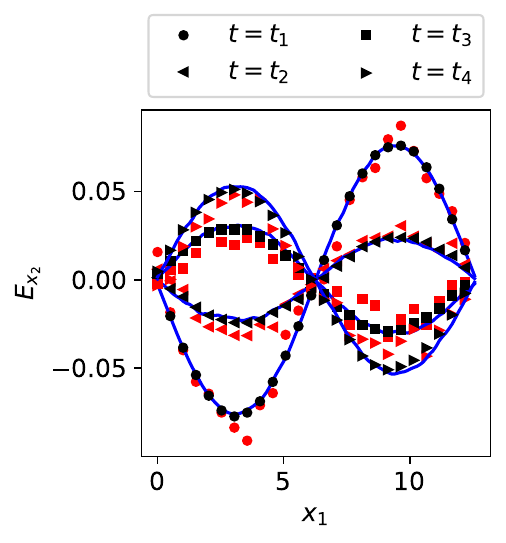}
  \\
  \hspace*{1cm} (a)  
  & 
  \hspace*{1cm} (b) 
  &
\hspace*{1cm} (c) 
 &
\hspace*{1cm} (d)
 \end{tabular}
   \caption{
  Evolution of electrostatic potential  $\phi$ on $x_1=L/2$ (a) and $x_2=L/2$ (c), as well as the electric field component $E_{x_1}$ on $x_1=L/2$ (b)  and  $E_{x_2}$ on $x_2=L/2$ (d) using PIC (red) and VR-PIC (blue) with $10^6$ particles as well as the reference PIC solution with $10^8$ (black), respectively, at $t_1,...,t_4=50\Delta t, 100\Delta t, 150\Delta t$ and $200\Delta t$.
  }
  \label{fig:landau_phi_E_on_line}
\end{figure}

\begin{figure}[H]
\hspace*{-1.5cm}
    \begin{tabular}{cccc}
    \includegraphics[scale=0.5]{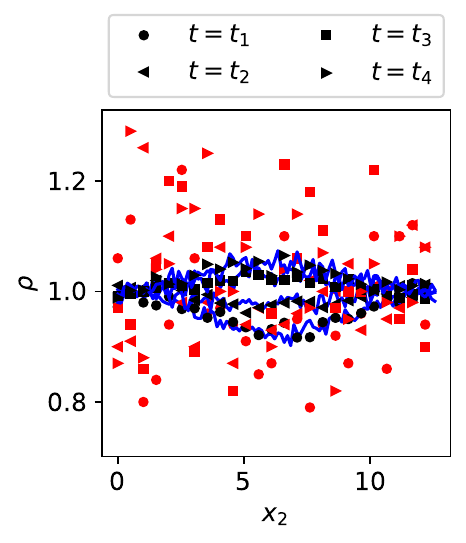}
   &
  \includegraphics[scale=0.5]{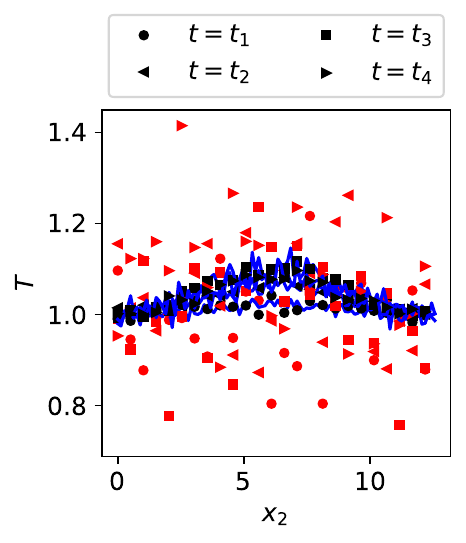}
   &
   \includegraphics[scale=0.5]{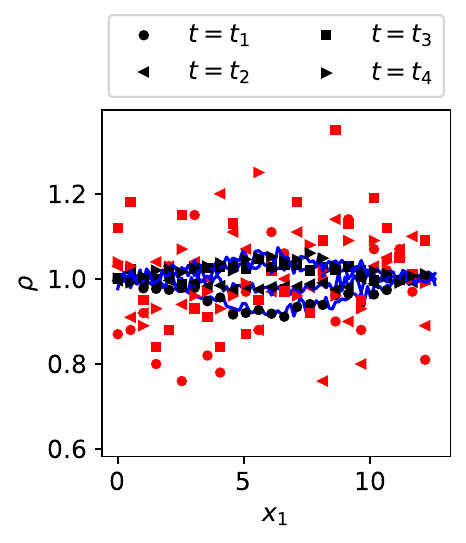}
   &
   \includegraphics[scale=0.5]{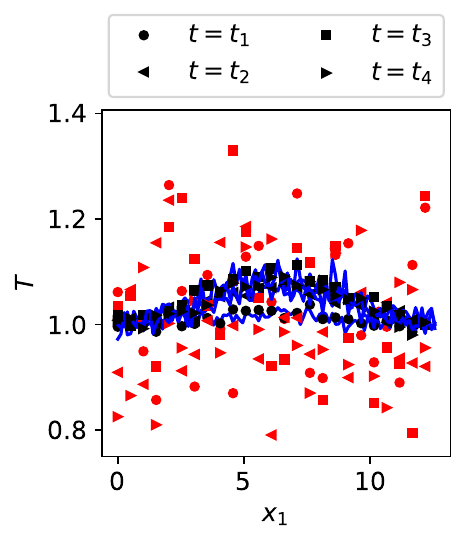}
  \\
   \hspace*{0.9cm}  (a)  
  & 
   \hspace*{0.9cm}  (b) 
  &
   \hspace*{0.9cm}  (c) 
 &
   \hspace*{0.9cm}  (d)
  \end{tabular}
  \caption{
 Evolution of charge density $\rho$ on $x_1=L/2$ (a) and $x_2=L/2$ (c), as well as the temperature $T$ on $x_1=L/2$ (b)  and on $x_2=L/2$ (d) using PIC (red) and VR-PIC (blue) with $10^6$ particles as well as the reference PIC solution with $10^8$ (black), respectively, at $t_1,...,t_4=50\Delta t, 100\Delta t, 150\Delta t$ and $200\Delta t$.
  } \label{fig:landau_rho_T_on_line}
\end{figure}

\begin{table}[H]
    \centering
    \caption{Comparison of computational cost between PIC and VR-PIC for the Landau Damping test case run for 3000 time steps and averaged over 5 ensemble.}
    \label{tab:computational_cost_landau_damping}
    \begin{tabular}{c c c c}
        \toprule
        $N_{p,\mathrm{tot}}$ &
        {Method} & {Execution Time [s]} \\
        \midrule
        $10^6$ & PIC &  34.5  
        \\
        & VR-PIC &  42.9  
        \\ \midrule
        $10^7$ & PIC & 80.7 
        \\
        & VR-PIC & 180.3 
        \\
        \midrule
        $10^8$ & PIC & 735.1 
        \\
        & VR-PIC & 1740.7  
         \\
        \bottomrule
    \end{tabular}
\end{table}

\noindent We further analyze the accuracy of the VR-PIC method using the analytical solution to the damping of the electric field in $L^2$-norm, i.e. $(\int |E|^2 dx)^{1/2}$, as the reference which is a commonly used benchmark in the literature \cite{kormann2015semi,kirchhart2024numerical}. 
Here, we perform a convergence study in terms of number of particles and  we use the same discretizations of the physical space to solve the Poisson equation, and deploy on-average $N_p=10^2, 10^3$ and $10^4$ particles per cell, i.e. in total $N_{p,\mathrm{tot}}=10^6,10^7$ and $10^8$ particles. Interestingly, as shown in Fig.~\ref{fig:damping_E2_convergance_Np}, 
the VR-PIC estimate of electric fields follows the theoretical damping for a longer time than the one computed using the PIC calculations, leading to a significant speedup. For example, the solution of PIC with $N_{p,\mathrm{tot}}=10^8$ particles is close to the PIC-VR solution with only $N_{p,\mathrm{tot}}=10^6$ particles achieving $\approx17$ times speed up with 100 times less memory footprint. As discussed in \S~\ref{sec:const_analysis}, this speed up increases quadratically as the magnitude of the signal, e.g. the perturbation $\alpha$ in the initial distribution eq.~\eqref{eq:landau_damping_dist}, becomes smaller. We report the execution times as a reference for the reader in Table~\ref{tab:computational_cost_landau_damping}.

\begin{figure}[H]
    \centering
    \includegraphics[width=1\linewidth]{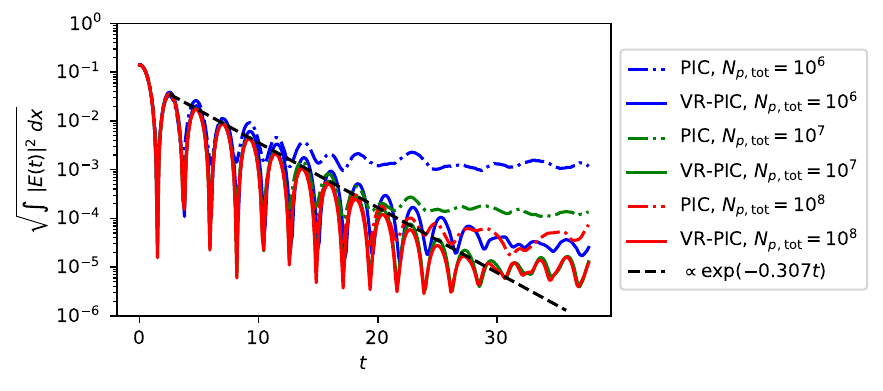}
    \caption{
    Convergence study on the damping of $ |E|$ in $L^2$-norm obtained with PIC and VR-PIC estimates using $10^6, 10^7$ and $10^8$ particles per cell on a  $100\times 100$ uniform mesh.
    }
    \label{fig:damping_E2_convergance_Np}
\end{figure}

\section{Conclusion}
\label{sec:conclusion}
\noindent In this paper, we presented a new conservative and entropic variance reduction method via importance weights with the local equilibrium distribution as the control variate for the particle-in-cell solution to the Vlasov-Poisson equation. 
Here, we observe that by keeping the weights constant during the kick, the evolution of importance weights remain stable in the low signal regime at the expense of introducing bias. In the introduced method, we propose a first-order correction to the weight distribution based on a semi-analytical solution to the relaxation of moments during kick. Then, as a mean to correct the weights after the kick, we deploy the maximum cross-entropy formulation which minimizes the bias while enforces the moment conservation at the particle level.
\\ \ \\
In several test cases, we analyzed the accuracy and efficiency of the proposed method. This includes Sod's shock tube in one and weak Landau Damping in two dimensional setting and made comparison to the analytical solution. We observe that the entropic VR-PIC achieves 1-4 orders of magnitude reduction in computational cost depending on the weakness of the signal while maintaining a high accuracy in the prediction. Given its success, we intend to extend the range of application of the entropic variance reduction method to other relevant kinetic equations such as electromagnetic Vlasov–Maxwell PIC in our future studies


\bibliographystyle{elsarticle-num} 
\bibliography{refs}




\end{document}